\documentclass[11pt]{article}

\usepackage{fullpage}
\usepackage{amsmath}
\usepackage{amsthm}
\usepackage{amssymb}
\usepackage{amsfonts}
\usepackage{hyperref}
\usepackage{color}
\usepackage{xcolor}

\usepackage{enumitem}
\setlist{itemsep=-3pt}

\newcommand{\ignore}[1]{}

\newtheorem{theorem}{Theorem}

\newtheorem{lemma}{Lemma}

\newtheorem{corollary}[lemma]{Corollary}

\renewcommand{\Pr}{{\bf Pr}}
\newcommand{\E}{{\bf E}}
\newcommand{\D}{{\cal D}}
\newcommand{\T}{{\cal T}}
\newcommand{\M}{{\cal M}}

\newcommand{\DT}{{\mathbb{DT}}}

\newcommand{\size}{{\rm psize}}

\renewcommand{\P}{\mathbb{P}}
\newcommand{\junta}{\mathbb{JUNTA}}

\newcommand{\mybox}[1]{\noindent\fbox{\parbox{\textwidth}{#1}}}

\begin{document}

\title{On Learning and Testing Decision Tree}
\author{{\bf Nader H. Bshouty\footnote{Part of this research was done when the first author was a visiting professor at Guangdong Technion Israel Institute of Technology (GTIIT), China.}}\hspace{.5in} 
{\bf Catherine A. Haddad-Zaknoon}\\
Dept. of Computer Science\\ Technion,  Haifa, 32000\\
}

\maketitle
\begin{abstract}
In this paper, we study learning and testing decision tree of size and depth that are significantly smaller than the number of attributes $n$.

Our main result addresses the problem of poly$(n,1/\epsilon)$ time algorithms with poly$(s,1/\epsilon)$ query complexity (independent of $n$) that distinguish between functions that are decision trees of size $s$ from functions that are $\epsilon$-far from any decision tree of size $\phi(s,1/\epsilon)$, for some function $\phi > s$. The best known result is the recent one that follows from Blank, Lange and Tan,~\cite{BlancLT20}, that gives $\phi(s,1/\epsilon)=2^{O((\log^3s)/\epsilon^3)}$. In this paper, we give a new algorithm that achieves  $\phi(s,1/\epsilon)=2^{O(\log^2 (s/\epsilon))}$.

Moreover, we study the testability of depth-$d$ decision tree and give a {\it distribution free} tester that distinguishes between depth-$d$ decision tree and functions that are $\epsilon$-far from depth-$d^2$ decision tree. In particular, for decision trees of size $s$, the above result holds in the distribution-free model when the tree depth is $O(\log(s/\epsilon))$. 

We also give other new results in learning and testing of size-$s$ decision trees and depth-$d$ decision trees that follows from{~\cite{BlancLT20,BlumerEHW87,Bshouty19,Bshouty20,BshoutyH19,EhrenfeuchtH89,GuijarroLR99,MehtaR02}} and some results we prove in this paper. 

\end{abstract}
\section{Introduction}
Decision tree is one of the popular predictive modelling approaches used in many areas including statistics, data mining and machine learning. Recently, learning and property-testing of sub-classes of decision trees have attracted much attention {~\cite{BlaisBM11,BlancLT20,Bshouty18,Bshouty20,BshoutyH19,ChakrabortyGM11,DiakonikolasLMORSW07,KearnsR00}}. In learning and property testing, the algorithm is provided by an access to a black box to some Boolean function $f$ and labeled random examples of $f$ according to some distribution $\D$. In property testing, given a subclass of decision trees $C$,  we need to decide whether $f$ is a decision tree in $C$ or ``far'' from being in $C$ with respect to $\D$,~\cite{BlumLR93,GoldreichGR98, RubinfeldS96}. On the other hand, in learning, it is guaranteed that $f\in C$ and we need to output a function  $g\in C$ close to $f$ with respect to the distribution $\D$,~\cite{Ang88,Valiant84}. 

\ignore{Given a subclass of decision trees $C$ and an access to a black box to a Boolean function $f$ and labeled random examples of $f$ according to some distribution $\D$. In property testing we need to decide whether $f$ is a decision tree in $C$ or ``far'' from being in $C$ with respect to $\D$,~\cite{BlumLR93,GoldreichGR98, RubinfeldS96}. In learning it is guaranteed that $f\in C$ and we need to output a function in $C$ close to $f$ with respect to distribution $\D$,~\cite{Ang88,Valiant84}.}

Since finding efficient algorithms for these problems is difficult and, in some cases, intractable, the following relaxation is considered. Let $H$ be a larger class of decision tree $H\supset C$. Then, in property testing, we are interested in the question: can we efficiently test $C$ by $H$? That is, to efficiently decide whether $f$ is a decision tree in $C$ or ``far'' from being in $H$ with respect to $\D$,~\cite{KearnsR00}. In learning theory, the question is: can we find an efficient algorithm that outputs $h\in H$ that is close to $f$ with respect to the distribution $\D$? The challenge in those problems is to find a small class $H\supset C$ such that efficient algorithms exist. In this paper, we address these problems as well as algorithms that are efficient in the query complexity that runs in quasi-polynomial time. 

\subsection{Models}

Let $C$ and $H\supseteq C$ be two classes of Boolean functions $f:\{0,1\}^n\to\{0,1\}$. In the {\it distribution-free} model, the algorithm has an access to a {\it black box query} and {\it random example query}. The black box query, for an input $x\in\{0,1\}^n$ returns $f(x)$. The random example query, when invoked, it returns a random example $(x,f(x))$ such that $x$  is chosen according to an arbitrary and unknown distribution~$\D$. In the {\it uniform distribution} model, $\D=U$ is the uniform distribution over $\{0,1\}^n$.

In the {\it distribution-free learning $C$ from $H$} (resp. uniform distribution learning), we need to construct, with high probability, via the above queries  to $f\in C$, a function $g\in H$ that is $\epsilon$-close to $f$ with respect to the distribution $\D$, i.e.,  $\Pr_\D[g(x)\not=f(x)]\le \epsilon$ (resp. where $\D=U$). The learning is called {\it proper learning} if $H=C$, and {\it exact learning} if $\epsilon=0$.

In the {\it distribution-free property testing $C$ by $H$}, \cite{KearnsR00}, (resp. uniform distribution property testing), for {\it any} Boolean function $f$, we need to distinguish, with high probability and via the above queries to $f$, between the case that $f$ is in $C$ versus the case that $f$ is $\epsilon$-far (not $\epsilon$-close) from every function in $H$ with respect to $\D$ (resp. the uniform distribution). Such an algorithm is called a {\it tester for $C$ by $H$}. When $H=C$,  the this algorithm is called a {\it tester for $C$}.

\subsection{Decision Tree}
A {\it decision tree} is a rooted binary tree in which each internal node is labeled with a variable $x_i$ and has two children. Each leaf is labeled with an output from $\{0,1\}$. A decision tree computes a Boolean function in the following way: given an input $x\in \{0,1\}^n$, the value of the function on $x$ is the output in the leaf reached by the path that starts at the root and goes left or right at each internal node according to whether the variable's value in $x$ is $0$ or $1$, respectively. 

The {\it size} of a decision tree is the number of leaves of the tree. The {\it depth} of a node (resp. leaf) in a decision tree is the number of edges in the path from the root to the node (resp. leaf). The depth of the tree is the maximum over all the depth values of its leaves. A {\it depth-$d$ decision tree} $T$ is a decision tree of depth at most~$d$. A {\it size-$s$ decision tree} $T$ is a decision trees of size at most $s$. 

We denote by $\DT_d$, $\DT^s$ and $\DT_d^s$ the classes of all depth-$d$ decision trees, size-$s$ decision trees and  depth-$d$ size-$s$ decision trees, respectively. Obviously, $\DT^s_d\subseteq \DT^s\subseteq \DT_s^s$ and $\DT^s_d\subseteq\DT_d\subseteq \DT_d^{2^d}$.

\subsection{Other Representations of Decision Tree}\label{DTSSS}\label{OR}
A {\it monomial} is a conjunction of variables,  and a {\it term} is a conjunction of literals (variable and negated variable). Two terms $t_1$ and $t_2$ are called {\it disjoint} if $t_1\wedge t_2=0$. A {\it multilinear polynomial} (or just a polynomial) is a sum (over the field $F_2$) of monomials. Every Boolean function can be expressed uniquely as a multilinear polynomial. A {\it disjoint-terms sum} is a sum of disjoint terms. Every Boolean function can be represented as a disjoint-terms sum. The representation is not unique.

A decision tree $f$ can be represented as a disjoint-terms sum according to the following recurrence. \ignore{To change a decision tree $f$ to a disjoint-terms sum we use the following recurrence.} If the decision tree is a leaf, then its disjoint-terms sum representation is the constant in this leaf. If the root label of $f$ is $x_i$ then, $f=x_if_1+\overline{x_i}f_0$ where $f_1$ and $f_0$ are the disjoint-terms sum of right and left sub-trees of $f$ respectively. It is easy to see that the number of terms we get in this recurrence is equal to the number of leaves in the tree labeled with $1$. Therefore, every leaf corresponds to a term and the number of literals in this term is equal to the depth of the leaf. To represent a disjoint-terms sum as a polynomial, we write for each appearance of $\overline{x_i}$ as $x_i+1$ and expand the expressions with the regular arithmetic rules in the field $F_2$.

\section{Main Result and Technique}
In this section, we present previous and new results. In the next two subsections, we consider two significant results in testing. Then in subsection~\ref{OResults} we give the other results. The tables in Figures~\ref{TABLE1},~\ref{TABLE2},~\ref{TABLE3} and~\ref{TABLE4} summarize all the results.

\subsection{Testing Decision Tree of Size $s$}
\ignore{For a function $f$, variables $x_{i_1},\ldots,x_{i_j}$ and $\xi_1,\ldots,\xi_j\in\{0,1\}$ we write $f_{|x_{i_1}\gets \xi_1,\ldots,x_{i_j}\gets \xi_j}$ for the function that results from substituting $x_{i_r}=\xi_r$, $r=1,\ldots,j$ in $f$.}
Let $f$ be a Boolean function over the variables $x_1, \ldots, x_n$. For $x_{i_1},\ldots,x_{i_j}$ and $\xi_1,\ldots,\xi_j\in\{0,1\}$, denote by $f_{|x_{i_1}\gets \xi_1,\ldots,x_{i_j}\gets \xi_j}$ the function that results from substituting $x_{i_r}=\xi_r$, $r=1,\ldots,j$ in $f$.

In~\cite{BlancLT20}, Blanc, Lange and Tan give the first tester that runs in poly$(n,1/\epsilon)$ time with poly$(\log s,$ $1/\epsilon)\log n$ query complexity and distinguishes between\footnote{Blanc, Lang and Tan solves a more challenging problem. Their tester is tolerant tester,~\cite{ParnasRS02}. That is, it distinguishes between functions that are $\epsilon$-close to $s$-size decision tree and functions that are $\Omega(\epsilon)$-far from every size-$\phi(s,1/\epsilon)$ decision tree.} functions that are $s$-size decision tree and functions that are $\epsilon$-far from every size-$\phi(s,1/\epsilon)$ decision tree for some function $\phi(s,1/\epsilon)$. When $n\gg s$, one can use the reduction from~\cite{Bshouty20} to get a poly$(n,1/\epsilon)$ time algorithm with poly$(s,1/\epsilon)$ query complexity (independent of $n$) for the same problem. The function $\phi$ achieved in~\cite{BlancLT20} is $\phi(s,1/\epsilon)=2^{(\log^3s)/\epsilon^3}$. We use a different approach to get $\phi(s,1/\epsilon)=2^{\log^2(s/\epsilon)}$. 

To achieve their result, Blanc et. al.,~\cite{BlancLT20}, define for every function $f$ a decision tree $T({d,f})$ of depth $d=O(\log^3s/\epsilon^3)$  as follows. They define the noise sensitivity  of $f$, NS$(f)$,  as the probability that $f(x)\not=f(x+y)$ where $x$ is uniform random, and for every $i$, $y_i=1$ with probability $p=\epsilon/\log s$. The score of $f$ with respect to $x_i$, Score$_i(f)$, is defined to be the expected decrease in the noise sensitivity of $f$ provided that $x_i$ is chosen to be the label in the root of the tree. The label of the root  of $T({d,f})$ is selected to be the variable $x_i$ that maximizes the score. The left and right sub-trees of $T(d,f)$ are $T(d-1,f_{|x_i\gets 0})$ and $T(d-1,f_{|x_i\gets 1})$, respectively. Then $T(0,g)$ is defined to be a leaf labeled with $0$ if $\E[g]<1/2$ and $1$ otherwise. Here $g$ is the function that results from $f$ by substituting the partial assignment defined by the path the leaf. They prove that, for $d=O(\log^3s/\epsilon^3)$, if $f$ is a size-$s$ decision tree, then $f$ is $\epsilon/4$-close to $T(d,f)$, and if $f$ is $\epsilon$-far from every size-$2^{O(\log^3s/\epsilon^3)}$ then, since $T(d,f)$ is size-$2^d$ $\left(=2^{O(\log^3s/\epsilon^3)}\right)$ decision tree, $f$ is $\epsilon$-far from $T(d,f)$. Moreover, they show that a query to $T(d,f)$ can be done in poly$(n,1/\epsilon)$ time and poly$(\log s,1/\epsilon)\log n$ queries. Therefore, $T(d,f)$ and $f$ can be queried to test if they are $\epsilon/4$-close or $\epsilon$-far. \ignore{ So they can query $T(d,f)$ and $f$ and test if they are $\epsilon/4$-close or $\epsilon$-far.} In~\cite{Bshouty20}, Bshouty gives a reduction that changes the query complexity of the tester from~\cite{BlancLT20} to be independent of $n$. This reduction gives a tester that solves the same problem in poly$(n,1/\epsilon)$ time and poly$(s,1/\epsilon)$ queries to $f$.

In this paper, we use a different approach. Let $f$ be a size-$s$ decision tree. Our algorithm regards $f$ as a polynomial. Although $f$ may have exponential number of monomials, we are interested in the influential ones only, that is, the small monomials. The number of small monomials in the polynomial representation of size-$s$ decision tree may be exponential. To control the number of small monomials, we shuffle the monomial by choosing a uniform random $a\in\{0,1\}^n$ and considering $T(x)=f(x+a)$. In disjoint-terms sum representation  of $f$ (see Subsection~\ref{DTSSS}), a large term $t$ in $f$ (a term that results from a leaf of depth $\Omega(\log(s/\epsilon))$ in the tree), with high probability (w.h.p), more than quarter of its variables become positive in $T$, and therefore, it only generates large monomials. Therefore, the shuffling process insures that $T$ has a small number, poly$(s/\epsilon)$, of significant monomials. This technique is used in~\cite{BshoutyM02} for learning decision tree under the uniform distribution. 

Let $c$ be a large constant and let $F$ be the sum of the monomials of size $r=c\log (s/\epsilon)$ in $T(x)$.  Let $G$ be the sum of monomials of size greater than $r$ and less than $16r$ in $T(x)$. Firstly, we run the algorithm of Bshouty and Mansour in~\cite{BshoutyM02}  to exactly learn $F$ and $G$ in polynomial time. If the learning algorithm fails, then w.h.p, $f$ is not size-$s$ decision tree and the algorithm rejects. Then, we define a decision tree $T(d,F,G)$ of depth $d=O(\log^2(s/\epsilon))$ as follows. Define Frac$(F,x_i)$ to be the fraction of the number of monomials in $F$ that contain $x_i$. Choose a variable $x_{i_1}$ with the minimum index $i_1$ that maximizes Frac$(F,x_{i_1})$ and use it as the label of the root of $T(d,F,G)$.  The left and right sub-trees of $T(d,F,G)$ are $T(d-1,F^{(0)},G^{(0)})$ and $T(d-1,F^{(1)},G^{(1)})$, respectively, where $F^{(\xi)}$ is the sum of all monomials that appear in $F_{|x_{i_1}\gets \xi}$ and not in $G_{|x_{i_1}\gets \xi}$, and $G^{(\xi)}$ is the sum of all monomials that appear in $G_{|x_{i_1}\gets \xi}$ and not in $F_{|x_{i_1}\gets \xi}$. 

We are interested in making a random walk in this tree. Assume we are given the first $m-1$ steps in the random walk $\xi^{(m-1)}=(\xi_1,\xi_2,\ldots,\xi_{m-1})\in\{0,1\}^{m-1}$, $F^{\xi^{(m-1)}}$ and $G^{\xi^{(m-1)}}$. We find the variable with the minimum index $i_m$ that maximizes Frac$(F^{\xi^{(m-1)}},x_{i_m})$. Choose a random uniform $\xi_{m}\in\{0,1\}$. Then, for $\xi^{(m)}=(\xi_1,\xi_2,\ldots,\xi_{m})$, $F^{\xi^{(m)}}$ is defined to be the sum of all the monomials in $F^{\xi^{(m-1)}}_{|x_{i_m}\gets \xi_m}$ that are not in $G^{\xi^{(m-1)}}_{|x_{i_m}\gets \xi_m}$, and $G^{\xi^{(m)}}$ is defined to be the sum of all the monomials in $G^{\xi^{(m-1)}}_{|x_{i_m}\gets \xi_m}$ that are not in $F^{\xi^{(m-1)}}_{|x_{i_m}\gets \xi_m}$. If $F^{\xi^{(m)}}$ is a constant function $\eta\in\{0,1\}$, then we have reached a leaf labeled with $\eta$. The way we define $F^{\xi^{(m)}}$ and $G^{\xi^{(m)}}$ turns to be  crucial in the algorithm and is needed for its proof of correctness. We note here that, unlike the tester of Blanc, Lange and Tan, this tree is not the decision tree of $T$ or $F$, because every path in the tree treats $F$ and $G$ as sets of monomials rather than functions.

We show that, when $f$ is a size-$s$ decision tree, for a random shuffling and random $\xi=(\xi_1,\xi_2,\ldots,$ $\xi_m)\in \{0,1\}$, with high probability,  Frac$(F^{\xi^{(m)}},x_{i_m})$ is at least $1/O(\log (s/\epsilon))$. Hence, for a random walk in the tree $T(d,F,G)$, each step decreases the number of monomials in $F^{\xi^{(m)}}$, on average, by a factor of $1-1/O(\log (s/\epsilon))$. Therefore, since $F$ contains at most $poly(s/\epsilon)$ monomials, with high probability, a random walk in $T(d,F,G)$ reaches a leaf in $O(\log^2(s/\epsilon))$ steps. \ignore{We then show that all the monomials that were ignored (those that are removed from $F$) are, with high probability, negligible. So this tree is $\epsilon/4$-close to $T$ if $f$ is size-$s$ decision tree.} 

Now suppose $f$ is $\epsilon$-far from every size-$2^{O(\log^2(s/\epsilon))}$ decision tree. It may happen that a function $f$ that is $\epsilon$-far from size-$2^{O(\log^2(s/\epsilon))}$ decision tree passes all the above tests, because the above algorithm relies only on the small monomials of $T$. Moreover, it may happen that the small monomials of such function coincide with the monomials of a small size decision tree. As a result, we add another test at each leaf of the tree. We test that the function $T$ at the leaf of the tree $T_{|x_{i_1}\gets\xi_1,\ldots,x_{i_m}\gets \xi_m}$ is $\epsilon/4$-close to a constant function. 

For a function that is $\epsilon$-far from every size-$2^{O(\log^2(s/\epsilon))}$ decision tree, if it is not rejected because its small monomials coincide with the monomial of small size decision tree, then the random walks will often reach a small depth leaf. On the other hand, if almost all the small depth leaves give a good approximation of the function, then $T$ is $\epsilon$-close to a small depth tree. Therefore, the function is rejected with high probability.

What if $f$ is size-$s$ decision tree and it does not pass this additional test? We show that with high probability this cannot happen. We show that, although the tester treats $F$ and $G$ as sets of monomials and not as functions, the tree gives a good approximation of $T$, which implies that almost all the leaves of small depth in the tree pass the last test. 

The above tester runs in $poly(n,1/\epsilon)$ time and queries. Using the reduction in~\cite{Bshouty20}, it can be changed to a tester that runs in $poly(s,1/\epsilon)n$ time and makes $poly(s,1/\epsilon)$ queries.  

\subsection{Testing Decision Tree of Depth $d$}
The algorithm of Blanc, Lange and Tan~\cite{BlancLT20} also distinguishes between depth-$d$ decision tree and functions that are $\epsilon$-far from depth-$O(d^3/\epsilon^3)$ decision trees under the uniform distribution. Again, we can make the query complexity independent of $n$ using the reduction of Bshouty,~\cite{Bshouty20}, and get a $2^{O(d)}n$ time uniform-distribution tester that asks $2^{O(d)}/\epsilon$ queries and distinguishes between depth-$d$ decision tree and functions that are $\epsilon$-far from depth-$O(d^3/\epsilon^3)$ decision trees. 

In this paper, we give a new simple {\it distribution-free} tester that runs in $2^{O(d)}n$ time, asks $2^{O(d)}/\epsilon$ queries and distinguishes between depth-$d$ decision tree and functions that are $\epsilon$-far from depth-$d^2$ decision trees.

Our algorithm relies on the following fact. Let $f$ be a depth-$d$ decision tree. Consider the polynomial representation of $f=M_1+M_2+\cdots+M_m$. For any maximal monomial $M_i=x_{i_1}x_{i_2}\cdots x_{i_t}$ (a monomial that is not sub-monomial of any other monomial in $f$)  and any $\xi_1,\ldots,\xi_t\in\{0,1\}$, the function $f_{|x_{i_1}\gets \xi_1,\ldots,x_{i_t}\gets \xi_t}$ is depth $(d-1)$-decision tree. 

We can define a depth-$d^2$ decision tree $T_f$ that is equivalent to $f$ as follows: Find a maximal monomial $M_i=x_{i_1}x_{i_2}\cdots x_{i_t}$. Then use all its variables in the first $t$ levels of the decision tree $T_f$. That is, build a complete tree of depth $t$ that all its nodes at level $j$ are labeled with $x_{i_j}$. This defines a different path for each $x_{i_1}= \xi_1,\ldots,x_{i_t}= \xi_t$. Then, in the last node of such path, attach the tree $T_g$ for  $g=f_{|x_{i_1}\gets \xi_1,\ldots,x_{i_t}\gets \xi_t}$. Since depth-$d$ decision trees are degree-$d$ polynomials, we have $t\le d$, and since the decision trees at level $t$ are depth-$(d-1)$ decision trees, the depth of $T_f$ is at most $d^2$ (in fact it is at most $d(d-1)/2$).

For the tester, we will not construct $T_f$, but instead, we show that for any assignment $a$, finding the route that $a$ takes in the tree $T_f$ can be done efficiently. For this end, we first show that if $f$ is a depth-$d$ decision tree then, the relevant variables of $f$ can be found in $\tilde O(2^{2d})+2^d\log n$ queries. Then, we show that a maximal monomial of $f$ can be found in $\tilde O(2^d)$ queries.  
For an assignment $a$ drawn according to a distribution $\D$, if $f$ is a depth-$d$ decision tree, the route that $a$ takes in $T_f$ ends before depth $d^2$. If $f$ is $\epsilon$-far from depth-$d$ decision tree, then, either finding the relevant variables of $f$ fails, or finding a maximal monomial of size at most $d$ fails or, with probability at least $\epsilon$, the route in $T_f$ goes beyond depth $d^2$. The later happens because if it does not for $O(1/\epsilon)$ examples drawn according to a distribution $\D$, then truncating the tree up to depth $d^2$, results a tree that is w.h.p $\epsilon$-close of a depth $d^2$-decision tree with respect to $\D$.

Notice that the query complexity of this tester depends on $n$ because finding the relevant variables of $f$ takes $\tilde O(2^{2d})+2^d\log n$ queries which depends on $n$. To make the query complexity independent of $n$, we use the reduction of Bshouty in~\cite{Bshouty20}.

\subsection{Other Results}\label{OResults}
The other results of the paper are summarized in the tables in Figures~\ref{TABLE1},~\ref{TABLE2},~\ref{TABLE3} and~\ref{TABLE4}, and they follow from the following known and new results:
\begin{enumerate}
\item\label{R1} A new reduction that changes any algorithm that learns a subclass of $k$-$\junta$ to a learning algorithm that makes number of queries  linear in $\log n$ and time linear in $n$. This reduction is similar to the one in~\cite{Bshouty19,Bshouty20}. See Lemma~\ref{Reduction}.
\item\label{R2} A new algorithm that finds a depth-$d$ size-$s$ decision tree that is consistent with a sample $S$. See Lemma~\ref{LearnDTds}.
\item\label{R3} The well known Occam's Razor lemma. See Lemma~\ref{occam}.
\item\label{R4} A natural generalization of the reductions in~\cite{Bshouty20} that changes a tester for a subclass of $k$-$\junta$ that has query complexity depends on $n$ to a tester with query complexity independent of $n$. See Lemmas~\ref{ThTriv2v2}, \ref{ThTriUv2} and~\ref{ThTriv}.
\item\label{R5} Blanc et. al. testers~\cite{BlancLT20}. See Lemma~\ref{Blan1} and \ref{Blan2}.
\item \label{R6} Other old and new results that are included in the context of the proofs. \ignore{proofs or before the proofs.} 
\end{enumerate}

In the table in Figure~\ref{TABLE1}, we give old and new results on proper learning of sub-classes of decision tree, and in Figure~\ref{TABLE2}, we give  old and new results on non-proper learning of sub-classes of decision tree.
In both tables, ${\cal M}=$D, U and E stand for uniform distribution, distribution-free and exact learning models, respectively. Moreover, $C$ is the class being learned, $\#${\bf RED}, $q$, is the number of the random example queries, $\#${\bf BBQ}, $m$, is the number of black box queries, and $T$ is the (computational) time complexity. The $+O((q+m)n)$ (in the first row of the table) is the time required for asking the queries. {\bf Ref} refers to the Theorem or the reference where the result is proved. The last column; {\bf Items}, indicates the items of the above list used in the proof of the result. Recall that, $\DT_d$, $\DT^s$ and $\DT_d^s$ the classes of all depth-$d$ decision trees, size-$s$ decision trees and  depth-$d$ size-$s$ decision trees, respectively.

\begin{figure}[h!]
\begin{center}
\begin{tabular}{|c|c|c|c|c|c|c|}
\hline
{\bf ${\cal M}$ } &{\bf $C$ }&  {\bf $\#$REQ}& {\bf $\#$BBQ} & {\bf Time}&{\bf Ref.}&{\bf Items}\\
&&$q=O(\ \cdot\ )$&$m=O(\ \cdot\ )$&$T=\cdot+O((q+m)n)$&&\\
\hline \hline
D & $\DT^s_d$ &  $ \frac{1}{\epsilon}\left(s\log n+\log\frac{1}{\delta}\right)$ & $-$ & $n^{O(d)}\frac{1}{\epsilon}\log\frac{1}{\delta}$ &Th.~\ref{DTds}.&\ref{R2}+\ref{R3}\\
\hline
D & $\DT_d^s$ & $\frac{s}{\epsilon}\log\frac{s}{\delta}$ & $s\log n$ & $s^{O(d)}\frac{1}{\epsilon}\log\frac{1}{\delta}+$ &Th.~\ref{DTds2}.&\ref{R1}+Th.~\ref{DTds}.\\
\hline
E & $\DT_d^s$ & $-$ & $s\left(2^d\log \frac{s}{\delta}+\log n\right)$  &$s^{O(d)}\log\frac{1}{\delta}+$ &Th.~\ref{EwE}&Th.~\ref{DTds2}.\\
\hline
E & $\DT_d^s$ & $-$ & $2^{2d+o(d)}\log n$  &$s^{O(d)}+$ &Th.~\ref{DTdsnew}&\ref{R2}+\ref{R6}\\
\hline\hline
D & $\DT^s$ & $\frac{s}{\epsilon}\log\frac{s}{\delta}$ & $s\log n$ & $s^{O(s)}+$ &Th.~\ref{DTSS}.&\ref{R1}+folklore\\
\hline
D & $\DT^s$ & $\frac{s}{\epsilon}\log\frac{s}{\delta}$ & $2^s+s\log n$ & $2^{O(s)}+$ &Th.~\ref{DTSE}.&\ref{R1}+\cite{GuijarroLR99}\\
\hline
D & $\DT^s$ & $\frac{1}{\epsilon}(s\log n+\log\frac{1}{\delta})$ & $-$&  $2^n\frac{s^2}{\epsilon^2}\log^2\frac{1}{\delta}$ &Th.~\ref{DTdsTPre}.&\ref{R2}+\ref{R3}\\
\hline
D & $\DT^s$ & $\frac{s}{\epsilon}\log\frac{s}{\delta}$ & $s\log n$ &  $\tilde O(2^{s})\frac{1}{\epsilon^2}\log^2\frac{1}{\delta}+$ &Th.~\ref{DTdsT}.&\ref{R1}+Th. \ref{DTdsTPre}\\
\hline
E & $\DT^s$ & $-$ & $2^n$ & $2^{O(n)}$ &\cite{GuijarroLR99}&\cite{GuijarroLR99}\\
\hline
E & $\DT^s$ & $-$ & $2^s\log\frac{1}{\delta}+s\log n$ & $2^{O(s)}+$ &Th.~\ref{DTSE2}.&\cite{BshoutyC18,GuijarroLR99}\\
\hline\hline
D & $\DT_d$ & $\frac{1}{\epsilon}\left(2^d\log n+\log\frac{1}{\delta}\right)$ & $-$ & $n^{O(d)}2^{O(d^2)}\frac{1}{\epsilon}\log\frac{1}{\delta}$ &Th.~\ref{DTds}.&$\DT_d=\DT_d^{2^d}$\\
\hline
D & $\DT_d$ & $ \frac{2^d}{\epsilon}\left(d+\log\frac{1}{\delta}\right)$ & $2^d\log n$ & $2^{O(d^2)}\frac{1}{\epsilon}\log\frac{1}{\delta}+$ &Th.~\ref{DTds2}.&$\DT_d=\DT_d^{2^d}$\\
\hline
E & $\DT_d$ & $-$ & $\tilde O(2^{2d})\log \frac{1}{\delta}+2^d\log n$  &$2^{O(d^2)}\log\frac{1}{\delta}+$ &Th.~\ref{EwE}&$\DT_d=\DT_d^{2^d}$\\
\hline
E & $\DT_d$ & $-$ & $2^{2d+o(d)}\log n$  &$2^{O(d^2)}+$ &Th.~\ref{DTdsnew}&$\DT_d=\DT_d^{2^d}$\\
\hline\hline
U & $\DT^s$ & $\frac{1}{\epsilon^2}\left(s\log n+\log\frac{1}{\delta}\right)$ & $-$ & $\left(\frac{s}{\epsilon}\right)^{O(\log n)}\log\frac{1}{\delta}$ &~\cite{MehtaR02}.&~\cite{MehtaR02}.\\
\hline
U & $\DT^s$ & $\frac{1}{\epsilon}\left(s\log n\log\frac{1}{\delta}\right)$ & $-$ & $\left(\frac{s}{\epsilon}\log n\right)^{O(\log n)}\log\frac{1}{\delta}$ &Th~\ref{DTsU}.&\ref{R2}+\ref{R3}\\
\hline
U & $\DT^s$ & $\frac{s}{\epsilon}\log s\log\frac{1}{\delta}$ & $s\log n$ & $\left(\frac{s}{\epsilon}\right)^{O(\log s)}\log\frac{1}{\delta}+$ &Th~\ref{DTsU2}.&\ref{R1}+Th. \ref{DTsU}. \\
\hline
\end{tabular}
\end{center}
	\caption{Results on proper learning decision tree.}
	\label{TABLE1}
	\end{figure}

\begin{figure}[h!]
\begin{center}
\begin{tabular}{|c|c|c|c|c|c|c|}
\hline
{\bf ${\cal M}$ } &{\bf $C/H$ }&  {\bf $\#$REQ}& {\bf $\#$BBQ} & {\bf Time}&{\bf Ref.}&Item\\
&$C$ from $H$&$q=O(\cdot)$&$m=O(\cdot)$&$T=O(\cdot)$&&\\
\hline
D&$\DT^s/\DT^{n^{O(\log s)}}$&$\frac{1}{\epsilon}\left(n^{O(\log s)}+\log\frac{1}{\delta}\right)$ & - &$qn^{O(\log s)}$&\cite{EhrenfeuchtH89}&\cite{EhrenfeuchtH89}\\ 
\hline
D&$\DT^s/\DT^{s^{O(\log s)}}$&$\frac{1}{\epsilon}\left(s^{O(\log s)}+s\log\frac{1}{\delta}\right)$ & $s\log n$ &$q(s^{O(\log s)}+\tilde O(n))$&Th.~\ref{DTsOUR}&1+\cite{EhrenfeuchtH89}\\ 
\hline 
\end{tabular}
\end{center}
	\caption{Results on non-proper learning decision tree.}
	\label{TABLE2}
	\end{figure}

In the two tables in Figure~\ref{TABLE3} and Figure~\ref{TABLE4}, we give old and new results on testing decision trees. All the notations are the same as the ones in the first two tables.

\begin{figure}[h!]
\begin{center}
\begin{tabular}{|c|c|c|c|c|c|}
\hline
{$\cal M$} &{\bf $C$ }&  {\bf $\#$Queries$=q$} & {\bf Time}$+qn$&{\bf Reference}&{\bf Items}\\
\hline \hline
D & $\DT^s_d$ &  $\tilde O\left(\frac{s^2}{\epsilon}\right)$ & $\frac{s^{O(d)}}{\epsilon}+$ &Th.~\ref{TH1}.&\ref{R4}+Th. \ref{DTds}.\\
\hline
D & $\DT^s_d$ &  $\tilde O\left(\frac{s^2}{\epsilon}\right)$ & ${\tilde O(s^{s})}+$ &\cite{Bshouty20} & \cite{Bshouty20} \\
\hline
D & $\DT^s_d$ &  $\tilde O\left(\frac{s^2}{\epsilon}\right)$ & ${\tilde O(2^{s})}+$ &Th.~\ref{TH1Uexp}.&\ref{R4}+Th.~\ref{DTdsTPre}.\\
\hline\hline
D & $\DT^s$ &  $\tilde O\left(\frac{s^2}{\epsilon}\right)$ & ${\tilde O(s^{s})}+$ &\cite{Bshouty20} & \cite{Bshouty20} \\
\hline
D & $\DT^s$ &  $\tilde O\left(\frac{s^2}{\epsilon}\right)$ & ${\tilde O(2^{s})}+$ &Th.~\ref{TH1Uexp}.&\ref{R4}+Th.~\ref{DTdsTPre}.\\
\hline\hline 
D & $\DT_d$ &  $\tilde O\left(2^{2d}+\frac{2^d}{\epsilon}\right)$ & $2^{O(d^2)}+$ &Th.~\ref{TH12}.&\ref{R4}+Th.~\ref{EwE}.\\
\hline\hline
U & $\DT^s_d$ &  $\tilde O\left(\frac{s}{\epsilon}\right)$ & $\frac{s^{O(d)}}{\epsilon}+$ &Th.~\ref{TH1U}.&\ref{R4}+Th. \ref{DTds}.\\
\hline
U & $\DT^s$ &  $\tilde O\left(\frac{s}{\epsilon^2}\right)$ & $\left(\frac{s}{\epsilon}\right)^{O(\log s)}+$ & \cite{MehtaR02} &\ref{R4}+\cite{MehtaR02}\\
\hline
U & $\DT^s$ &  $\tilde O\left(\frac{s}{\epsilon}\right)$ & $\left(\frac{s}{\epsilon}\right)^{O(\log s)}+$ & Th.~\ref{TH1UB}. &\ref{R4}+Th.\ref{DTsU}.\\
\hline
U & $\DT_d$ &  $\tilde O\left(\frac{2^{d}}{\epsilon}\right)$ & $\frac{2^{O(d^2)}}{\epsilon}+$ &Th.~\ref{TH12U}.&$\DT_d=\DT_d^{2^d}$\\
\hline
\end{tabular}
\end{center}
	\caption{Results on testing decision tree.}
	\label{TABLE3}
	\end{figure}

\begin{figure}[h!]
\begin{center}
\begin{tabular}{|c|c|c|c|c|c|}
\hline
$\cal M$ &{\bf $C$ }&  {\bf $H=\DT^S,\DT_D$ }&{\bf $\#$Queries } &{\bf Reference}&{\bf Item}\\
\hline \hline
D & $\DT^s$ & $S={s^{O(\log s)}}$ & $ \frac{s^{O(\log s)}}{\epsilon}$  &Th. \ref{TH1NP}&\ref{R1}+\ref{R4}+\cite{EhrenfeuchtH89}\\
\hline\hline
U & $\DT^s$ & $S={s^{O\left(\frac{\log^2 s}{\epsilon^3}\right)}}$ & $ poly(\log s,\frac{1}{\epsilon})+\tilde O\left(\frac{s}{\epsilon}\right)$  &Th. \ref{BlankTh1}.&\ref{R4}+\cite{BlancLT20}\\
\hline
U & $\DT^s$ & $S=(s/\epsilon)^{O(\log (s/\epsilon))}$ & $poly(s,1/\epsilon)$ &Th. \ref{TH44}.& Ours\\
\hline\hline
U & $\DT_d$ & $D={{\frac{d^3}{\epsilon^3}}}$ & $ poly(d,\frac{1}{\epsilon})+\tilde O\left(\frac{2^d}{\epsilon}\right)$ &Th. \ref{BlankTh2}.&\ref{R4}+\cite{BlancLT20}\\
\hline 
D & $\DT_d$ & $D=d^2$ & $ \tilde O\left(\frac{2^{2d}}{\epsilon}\right)$ &Th. \ref{TH33N}.& Ours\\

\hline
\end{tabular}
\end{center}
	\caption{Results on testing decision tree $C$ by larger decision trees $H$. The time here is $poly(q)n$ where $q$ is the number of queries.}
	\label{TABLE4}
	\end{figure}

\ignore{
\color{red}

For a Boolean function $f$ and a class $C$, when we write $f\in C$ we mean that there is $T\in C$ that is equivalent to $f$. For two classes $C_1$ and $C_2$ we write $C_1\subseteq C_2$ if for every $f\in C_1$ there is $g\in C_2$ that is equivalent to $f$. The class $C(m)$ is the class of all functions in $C$ that has at most $m$ relevant variables. 

We have

\begin{lemma}\label{SDT} We have $\DT_d\subseteq \P_d^{(2\cdot 3^{d-1})}(2^{d}-1)$. That is, 
every depth-$d$ decision tree is equivalent to some  $(2\cdot 3^{d-1})$-sparse degree-$d$ polynomial with at most $2^{d}-1$ relevant variables.
\end{lemma}
\begin{proof}
The proof is by induction. For decision tree of depth at most $1$ the function is $x_i$, $x_i+1$, $0$ or $1$. Let $T$ be a depth-$d$ decision tree. Let $x_{i_1}$ be the variable at the root. Let $T_0$ and $T_1$ be the left and right sub-trees of the root. Then $T=x_{i_1}T_1+\overline{x_{i_1}}T_0=x_{i_1}(T_1+T_0)+T_0$. Since $T_0$ and $T_1$ are of depth at most $d-1$, by the induction hypothesis, the result follows.
\end{proof}
\color{black}}

We organize the paper as follows. In Section~\ref{LDT}, we study proper and non-proper learning of sub-classes of decision trees. The section starts with some notations (subsection~\ref{N}), then gives a reduction in learning (subsection~\ref{ARiL}), proves the results in distribution-free proper learning (subsection~\ref{DFP}),  the results in uniform distribution learning (subsection~\ref{UD}) and then the results in non-proper learning (subsection~\ref{NpL}). See Figure~\ref{TABLE1} and Figure~\ref{TABLE2}.

In Section~\ref{TDT}, we study testing of sub-classes of decision trees. The section starts with introducing the natural generalization of the reductions in~\cite{Bshouty20} (subsection~\ref{RiT}) and the results of Blanc et. al.~\cite{BlancLT20} in testing (subsection~\ref{BLTT}). We then prove the results for distribution-free testing (subsection~\ref{DFT}), uniform distribution testing (subsection~\ref{UD2}) and testing by large classes of decision trees (subsection~\ref{TbLC}). See Figure~\ref{TABLE3}.

In Section~\ref{ATfDTd}, we give a distribution-free tester for the class of depth-$d$ decision tree and in Section~\ref{ATfDTs} we give a uniform distribution tester for the class of size-$s$ decision tree.

\section{Learning Decision Tree}\label{LDT}
In this section, we present all known results of learning decision tree, and we give new results. These results are summarized in Figure~\ref{TABLE1} and~\ref{TABLE2}.

\subsection{Notations} \label{N}
For a function $f$, variables $x_{i_1},\ldots,x_{i_j}$ and $\xi_1,\ldots,\xi_j\in\{0,1\}$, we write $f_{|x_{i_1}\gets \xi_1,\ldots,x_{i_j}\gets \xi_j}$ to denote the function that is derived from  substituting $x_{i_r}=\xi_r$, $r=1,\ldots,j$ in $f$. For a set of variables $X=\{x_{i_1},x_{i_2},\ldots,x_{i_j}\}$, we denote $f_{|X\gets 0}=f_{|x_{i_1}\gets 0,\ldots,x_{i_j}\gets 0}$.
We say that $x_i$ is {\it relevant} variable in $f$ if $f_{|x_i\gets 0}\not= f_{|x_i\gets 1}$. For $u\in\{0,1\}^n$, we denote $u_{|x_{i_1}\gets \xi_1,\ldots,x_{i_j}\gets \xi_j}$ (resp. $u_{|X\gets 0}$) the vector $v\in\{0,1\}^n$ 
such that $v_{i_\ell}=\xi_{i_\ell}$ (resp. $v_{i_\ell}=0$) for $\ell\in[j]$ and $v_i=u_i$, otherwise.

The class $k$-{$\junta$} is the class of Boolean functions that has at most $k$ relevant variables. Obviously, $\DT_d^s,\DT^s\subseteq s$-$\junta$ and $\DT_d\subseteq 2^d$-$\junta$.

We say that a class $C$ is {\it closed under zero-one projection} (resp. {\it zero projection}) if for every $f\in C$, every $i\in [n]$ and every $\xi\in\{0,1\}$ we have $f_{|x_i\gets \xi}\in C$ (resp. $f_{|x_i\gets 0}\in C$). We say that $C$ is {\it symmetric} if for every permutation $\pi:[n]\to [n]$ and every $f\in C$ we have $f_\pi\in C$ where $f_\pi(x):=f(x_{\pi(1)},\cdots,x_{\pi(n)})$. All the classes in this paper are closed under zero-one projection and symmetric. 

\subsection{A Reduction in Learning}\label{ARiL}
In this section, we give a reduction that changes any algorithm that learns a subclass of $k$-$\junta$ to a learning algorithm that makes a number of queries linear in $\log n$ and time linear in $n$. The reduction is similar to the one that is used in~\cite{Bshouty19,Bshouty20} for the distribution-free property testing of subclasses of $k$-$\junta$.

\begin{lemma}\label{Reduction}
Let $C\subseteq H$ be a sub-classes of $k$-$\junta$ that are closed under zero projection. Let ${\cal A}$ be a learning algorithm that learns $C$ with $M(n, \epsilon, \delta)$ black box queries and $Q(n,\epsilon,\delta)$ random example queries in time $T(n,\epsilon,\delta)$ and outputs a hypothesis in $H$ of size $S(n,\epsilon)$. Then, there is an algorithm ${\cal B}$ that learns $C$ with $$M'=M(k,\epsilon/2,\delta/2)+O(k\log n)$$ black box queries and $$Q'=Q(k,\epsilon/2,\delta/2)+O\left(\frac{k}{\epsilon}\log \frac{k}{\delta}\right)$$ random example queries in time $T(k,\epsilon/2,\delta/2)+(M'+Q')n$ and outputs a hypothesis in $H$ of size $S(k,\epsilon/2)$.
\end{lemma}

Before proving this result, consider the following procedure {\bf Find-Close}. We have (see \cite{Bshouty19}, Lemma 9 and 10):
\begin{lemma}{\rm \cite{Bshouty19}}.\label{cloose} Let $c>1$ be any constant. Let $f\in k$-$\junta$.
{\bf Find-Close} makes $O((k\log (k/\delta))/\epsilon)$ random example queries and $O(k\log n)$ black box queries to $f$ and with probability at least $1-\delta$ outputs a set of variables $X$ of size at least $n-k$ such that
$$\Pr_{u\in {\cal D}}[f_{|X\gets 0}(u)\not= f(u)]\le \epsilon/c.$$
\end{lemma}

\mybox{ {\bf Find-Close$(f)$}\\ \hspace{.3in}
{\bf Input}: black box access to $f$\\
{\bf Output}: A set of variables $X$ of size at least $n-k$ such that, with probability at least $1-\delta$, $h=f_{|X\gets 0}$ is $(\epsilon/c)$-close to $f$ with respect to the distribution $\D$.
\begin{enumerate}
    \item Set $X=\{x_i|i\in [n]\}$; $t_X=0$. 
    \item Repeat $M=ck(\ln (k/\delta))/\epsilon$ times.
    \item \hspace{.25in} Choose $u\in {\cal D}$.
    \item \hspace{.25in} $t_X\gets t_X+1$.
    \item \hspace{.25in} If $f_{|X\gets 0}(u)\not=f(u)$ then
    \item \hspace{.5in} Binary search from $u$ to $u_{|X\gets 0}$ to find a new relevant variable $x_\ell$ of $f$. 
    \item \hspace{.5in}
    $X\gets X\backslash \{x_\ell\}$.
    \item \hspace{.5in}
    $t_X\gets 0.$
    \item \hspace{.25in} If $t_X=c(\ln(k/\delta))/\epsilon$ then output $X$.
\end{enumerate}}
\\

Using the procedure {\bf Find-Close}, we prove Lemma~\ref{Reduction}.
\begin{proof}
 Algorithm ${\cal B}$ proceeds as follows. First, it runs {\bf Find-Close} with $c=2$ and success probability $1-\delta/2$. Then, it runs ${\cal A}$  with $k$ variables to learn $f_{|X\gets 0}$ with accuracy $\epsilon/2$ and success probability $1-\delta/2$. Let $g$ be the output of algorithm ${\cal B}$. By Lemma~\ref{cloose}, with probability at least $1-\delta/2$, $f_{|X\gets 0}$ is $\epsilon/2$-close to $f$ with respect to $\D$.  With probability at least $1-\delta/2$, $g$ is $\epsilon/2$-close to $f_{|X\gets 0}$ with respect to $\D$. Therefore, with probability at least $1-\delta$, $g$ is $\epsilon$-close to $f$ with respect to $\D$.
\end{proof}

\subsection{Distribution-Free Proper Learning}\label{DFP}
In this subsection, we give some known and new results on proper learning of decision tree.

For a {\it sample} $S\subseteq \{0,1\}^n\times \{0,1\}$, we say that a Boolean function $f$ {\it is consistent with} $S$ if for every $(x,y)\in S$ we have $f(x)=y$. We now prove the following:
\begin{lemma}\label{LearnDTds}
There is a $$t(n,d,r)=O\left(r\sum_{i=0}^d{n\choose i}\min(r,2^i)\right)$$ time algorithm that takes as an input, a sample 
$S\subseteq \{0,1\}^n\times \{0,1\}$ of size $r$ and returns a depth-$d$ size-$s$ decision tree $f$ that is consistent with $S$ if one exists, and {\rm NONE} otherwise. 

The algorithm can return the smallest size $f\in \DT_d^s$ or the smallest depth one. 
\end{lemma}
\begin{proof}
Consider the following algorithm {\bf Consis}$(S,n,d)$ that returns a tree of depth $d$ of minimum size that is consistent with $S$. Algorithm {\bf Consis}$(S,n,d)$ proceeds as follows. If some constant function $\xi\in\{0,1\}$ is consistent with $S$, then the decision tree is a leaf labeled with $\xi$. If no $\xi\in\{0,1\}$ is consistent with $S$ and $d=0$, then the algorithm returns NONE. Otherwise, assume that the algorithm knows that the variable at the root of the tree is $x_i$. Let   $S_{|x_i\gets 0}:=\{(a,b)\in S:a_i=0\}$ denote the sample of the left sub-tree, and similarly, let $S_{|x_i\gets 1}:=\{(a,b)\in S:a_i=1\}$  denote the samples of the right sub-tree. Then, recursively calculate $T_{i,0}\gets${\bf Consis}$(S_{|x_i\gets 0},n-1,d-1)$ and $T_{i,1}\gets${\bf Consis}$(S_{|x_i\gets 1},n-1,d-1)$. The algorithm runs the above for every $i$, and returns a tree with a root labeled with $x_j$ with the minimum $size(T_{j,1})+size(T_{j,2})$. To handle the case when the algorithm returns NONE, we assume that $size($NONE$)=\infty$.

If the number of calls to {\bf Consis} is $T(n,d)$ then, since we run the algorithm for every variable $x_i$, we have $T(s,d)\le n (T(n-1,d-1)+T(n-1,d-1))\le 2n T(n-1,d-1)$. Splitting $S$ at each call takes time at most $r$. This gives time complexity $O(r(2n)^d).$ To get the complexity in the Lemma we do the following.

We run {\bf Consis} backward from the leaves of the recursion calls to the root. The leaf of the recursion calls are {\bf Consis}$(S_{|x_{i_1}\gets \xi_1,\ldots,x_{i_d}\gets \xi_d},n-d,0)$ for all $1\le i_1<i_2<\cdots<i_d\le n$ and $\xi_1,\ldots,\xi_d\in\{0,1\}$ for which $S_{|x_{i_1}\gets \xi_1,\ldots,x_{i_d}\gets \xi_d}\not=\O$. The number of calls ${n\choose d}\min(r,2^d)$ and at level $i$ is~${n\choose i}\min(r,2^i)$. 

This finds the tree in $\DT_d^s$ with a minimum size. To find the tree in $\DT_d^s$ with minimum depth, we run {\bf Consis}$(S,n,j)$ for $j\in [d]$.
\end{proof}

A {\it random sample of $f$ according to a distribution $\D$} is a sample $S$ where $\{x|(x,y)\in S\}$ are independent and identically distributed according to $\D$ and $y=f(x)$. 

The next lemma is Occam's Razor Lemma,~\cite{BlumerEHW87}.
\begin{lemma}\label{occam} {\rm \cite{BlumerEHW87}} (Occam's Razor). Let $C$ and $H\supseteq C$ be any finite classes of Boolean functions. A consistent hypothesis  $h\in H$ with a set of random samples $S$ of $f\in C$ according to a distribution $\D$ of size 
$$|S|=O\left(\frac{1}{\epsilon}\left(\log|H|+\log\frac{1}{\delta}\right)\right)$$ is, with probability at least $1-\delta$, $\epsilon$-close to $f$ according to the distribution $\D$. 
\end{lemma}

We now can prove the first result.

\begin{theorem}\label{DTds} The class $\DT_d^s$ is  distribution-free properly learnable from $q=O((1/\epsilon)(s\log n+\log(1/\delta)))$ random example queries in time $qn^{O(d)}$.  
\end{theorem}
\begin{proof}
The proof follows from Lemma~\ref{LearnDTds}, Lemma~\ref{occam} and the fact that $|\DT_d^s|\le |\DT^s|\le  n^{O(s)}$.
\end{proof}
By Lemma~\ref{Reduction} with $k=s$ and Theorem~\ref{DTds}, we have:
\begin{theorem}\label{DTds2} The class $\DT_d^s$ is  distribution-free properly learnable using $q=O((s/\epsilon)\log (s/\delta)))$ random example queries and $m=O(s\log n)$ black box queries in time $s^{O(d)}(1/\epsilon)\log(1/\delta)+O((q+m)n)$.
\end{theorem}

Since $\DT_d=\DT_d^{2^d}$, Theorem~\ref{DTds} and~\ref{DTds2} give learning algorithms for $\DT_d$ (with $s=2^d$). 

We now show how to change the learning algorithm to (randomized) exact learning algorithm.
\begin{theorem}\label{EwE}
The class $\DT_d^s$ is exactly properly learnable from $m=O(s(2^d\log (s/\delta)+\log n))$ black box queries in time $s^{O(d)}\log(1/\delta)+O(mn)$.
\end{theorem}
\begin{proof}
We run the algorithm in Theorem~\ref{DTds2} with $\epsilon=1/2^{d+2}$ and uniform distribution. Since the learning is proper and for every two decision trees $T_1$ and $T_2$, $\Pr[T_1\not=T_2]\ge 1/2^d$, the result follows.
\end{proof}

We now give a deterministic exact proper learning algorithm for $\DT_d^s$. We  first define $(n,d)$-universal set. 

An $(n,d)$-universal set is a set $U\subseteq \{0,1\}^n$ such that for every $1\le i_1<\cdots<i_d\le n$ and every $\xi_1,\xi_2,\ldots,\xi_d\in\{0,1\}^n$, there is an assignment $a\in U$ such that $a_{i_j}=\xi_j$ for every $j\in[d]$. An $(n,d)$-universal set of size $2^{d+o(d)}\log n$ can be constructed in deterministic time\footnote{Here and in the proof of Theorem~\ref{DTdsnew}, $o(d)$ is $\log^2d$.} $2^{O(d)}n\log n$,~\cite{NaorSS95}. 

Using $(n,d)$-universal sets we show:
\begin{theorem}\label{DTdsnew}
The class $\DT_d^s$ is deterministically exactly properly learnable from $m=2^{d+o(d)}\log n$ black box queries in time $s^{O(d)}+O(mn)$.
\end{theorem}
\begin{proof}
We construct an $(n,2d)$-universal set $U$ of size $2^{2d+o(d)}\log n$ in time $2^{O(d)}n\log n$, find the relevant variables $V$ and then use the algorithm in Lemma~\ref{LearnDTds} to construct a decision tree in $\DT^s_d$ over the variables $V$ that is consistent with $S=\{(a,f(a))|a\in U\}$ in time $s^{O(d)}$. To finish the proof we show: (1) how to find the relevant variables and (2) we show that there is only one function in $\DT_d^s$ that is consistent with $S$.

For (1), in~\cite{BshoutyH19}, Bshouty and Haddad-Zaknoon give a deterministic algorithm that finds the relevant variables in deterministic time $2^{O(d)}n\log n$ using $2^{2d+o(d)}\log n$ queries. 

To prove (2). For the sake of contradiction, assume that there are two distinct functions $f_1,f_d\in \DT_d^s$ that are consistent with~$S$. 
\ignore{To this end, let, for the contrary, $f_1,f_d\in \DT_d^s$ be two distinct functions that are consistent with~$S$.} Since $0\not=g:=f_1+f_2\in \DT_{2d}$, there are $1\le i_1<\cdots<i_{2d}\le n$ and  $\xi_1,\xi_2,\ldots,\xi_{2d}\in\{0,1\}^n$ such that $g_{|x_{i_1}\gets \xi_1,\ldots,x_{i_{2d}}\gets \xi_{2d}}=1$. Since $U$ is an $(n,2d)$-universal set, there is an assignment  $a\in U$ such that $a_{i_j}=\xi_j$ for all $j\in[2d]$. Therefore, $g(a)=1$ which implies that $f_1(a)\not=f_2(a)$. Therefore, one of the functions $f_1$ or $f_2$ is not consistent with $S$ which is a contradiction.

\end{proof}

The following Theorem is a simple exhaustive search with Lemma~\ref{Reduction} 
\begin{theorem}\label{DTSS} The class $\DT^s$ is  distribution-free properly learnable  using $q=O((s/\epsilon)\log (s/\delta)))$ random example queries and $m=O(s\log n)$ black box queries in time $s^{O(s)}+O((m+q)n)$.
\end{theorem}
\begin{proof}
Consider the algorithm that takes a sample of size $(1/\epsilon)(s\log n+\log(1/\delta))$ and exhaustively search for a consistent decision tree of size $s$. The number of such decision trees is $|\DT^s|=n^{O(s)}$. Now, we use Lemma~\ref{Reduction} and get the result.
\end{proof}

The following is the only proper learning that is known for $\DT^s$ and runs in time less than $s^{O(s)}$
\begin{theorem}\label{DTSE}\cite{GuijarroLR99,MehtaR02}
The classes $\DT^s$ and $\DT_d^s$ are exactly properly learnable from $q=O((s/\epsilon)\log(s/\delta))$ random example queries and $m=O(2^{s} +s\log n)$ black box queries in time $2^{O(s)} + O((m+q)n)$.
\end{theorem}
\begin{proof}
A decision tree of minimum size for any function $f$  can be constructed from its truth-table ($2^n$ black box queries) in time $2^{O(n)}$. See for example~\cite{GuijarroLR99}. Now by Lemma~\ref{Reduction}, the result follows for $\DT^s$. 

To prove the result for $\DT_d^s$, use the algorithm in~\cite{MehtaR02} that constructs a size-$s$ depth-$d$ decision tree from a truth-table in time $2^{O(n)}$. By Lemma~\ref{Reduction}, the result follows.
\end{proof}

Our target is to improve the query complexity of the above two results.To that end,  we first prove:
\begin{theorem}\label{DTdsTPre} The class $\DT^s$ is  distribution-free properly learnable using $q=O((1/\epsilon)(s\log n+\log (1/\delta)))$ random example queries in time $O(q^22^n)$ time.
\end{theorem}
\begin{proof}
By Lemma~\ref{LearnDTds}, Lemma~\ref{occam}, the fact that $\DT^s=\DT_s^s$ and the fact that $|\DT^s|\le  n^{O(s)}$ we get an algorithm that properly learns $\DT^s$ from $q=O((1/\epsilon)(s\log n+\log(1/\delta)))$ random example queries and time $O(q^22^n)$. 
\end{proof}

Using Theorem~\ref{DTdsTPre} with Lemma~\ref{Reduction}, we get the following result: 
\begin{theorem}\label{DTdsT} The class $\DT^s$ is  distribution-free properly learnable using $q=O((s/\epsilon)\log (s/\delta)))$ random example queries and $m=O(s\log n)$ black box queries in time $\tilde O(2^s)(1/\epsilon^2)\log^2(1/\delta)+O((q+m)n)$ .
\end{theorem}

Moreover, for exact learning, the following two results can be achieved.
\begin{theorem}\label{DTSE2}
The class $\DT^s$ and $\DT_d^s$ are exactly properly learnable from $m=O(2^{s}\log(1/\delta) +s\log n)$ black box queries in time $2^{O(s)} + O(mn)$.
\end{theorem}
\begin{proof}
The learning algorithm for $s$-$\junta$ in~\cite{BshoutyC18} learns the decision tree from $s$-$\junta$ using $m=O(2^{s}\log(1/\delta) +s\log n)$ black box queries in time $O(mn)$. The algorithm also finds the relevant variables. Then, the minimal size depth-$d$ decision tree can be constructed from its truth-table in $2^s$ queries and time $2^{O(s)}$~\cite{GuijarroLR99}.
\end{proof}

The following is a simple and trivial reduction for decision tree learning:
\begin{lemma}\label{Red2} We have,
\begin{enumerate}
    \item If $\DT_d$ is uniform-distribution properly learnable from $Q(\epsilon,\delta)$ random example queries and $M(\epsilon,\delta)$ black box queries in time $T(\epsilon,\delta)$, then it is exactly properly learnable from $Q(1/2^{d+1},\delta)$ random example queries and $M(1/2^{d+1},\delta)$ black box queries in time $T(1/2^{d+1},\delta)$. 
    \item If $\DT_d$ is distribution-free properly learnable from $Q(\epsilon,\delta)$ random example queries and $M(\epsilon,\delta)$ black box queries in time $T(\epsilon,\delta)$, then it is exactly properly learnable from $Q(1/2^{d+1},\delta)+M(1/2^{d+1},\delta)$ black box queries in time $T(1/2^{d+1},\delta)$. 
\end{enumerate}
\end{lemma}
\begin{proof}
We just run the algorithm that uniform-distribution properly learns $\DT_d$ with $\epsilon=1/2^{d+1}$. Since the algorithm, with probability at least $1-\delta$, outputs a depth-$d$ decision tree $T$, and since $\Pr[T\not=f]\le 1/2^{d+1}$, we get that $T$ is equivalent to $f$. This follows from the fact that a depth-$d$ decision tree is a multivariate polynomial over the binary field of degree $d$k and for any two such non-equivalent polynomials $g$ and $h$, we have $\Pr[g\not=h]=\Pr[g+h\not=0]\ge 1/2^d.$

If $\DT_d$ is distribution-free properly learnable, then in particular, it is uniform-distribution properly learnable. Therefore, we can run the algorithm with $\epsilon=1/2^{d+1}$ and simulate each random example queries by asking black box query to a random and uniform $x$.
\end{proof}

We can use Theorem~\ref{DTds2} with Lemma~\ref{Red2} to get the following result on exact proper learning for $\DT_d(=\DT_d^{2^d})$. This also follows from Theorem~\ref{EwE}.
See also~\cite{BshoutyH19}, Corollary 10.\footnote{The algorithm in~\cite{BshoutyH19} returns a function equivalent to $f$ in a Fourier Representation of size $d4^d$. A depth-$d$ decision tree can be constructed from $f$ in time $2^{O(d^2)}$.}
\begin{corollary}\label{DT} The class $\DT_d$ is exactly properly learnable using $m=\tilde O(2^{2d})\log({1}/{\delta})+O(2^d\log n)$ black box queries in time $2^{O(d^2)}\log(1/\delta)+O(mn)$. 
\end{corollary}

\subsection{Uniform Distribution Proper Learning}\label{UD}
In this section, we give results for learning $\DT^s$ under the uniform distribution.

The following result is due to~\cite{MehtaR02}.
\begin{lemma}{\rm \cite{MehtaR02}}. The class $\DT^s$ is (uniform-distribution) properly learnable using $O((1/\epsilon^2)(s\log n+\log(1/\delta))$ random example queries in time $(s/\epsilon)^{O(\log n) }$ $\log(1/\delta)$. 
\end{lemma}

We will prove the following
\begin{theorem}\label{DTsU} The class $\DT^s$ is (uniform-distribution) properly learnable using $O((1/\epsilon)(s\log n\log(1/\delta)))$ random example queries in time $((s/\epsilon)\log n)^{O(\log n)}$ $\log(1/\delta)$.   
\end{theorem}
\begin{proof}
We prove that for any $\epsilon$, there is a learning algorithm that makes $m=O((1/\epsilon)s\log n)$ random example queries with success probability at least $1/2$. Then, running this algorithm $\log(2/\delta)$ times with accuracy $\epsilon/10$, with probability at least $1-\delta/2$, one of the hypothesis is $\epsilon/10$-close to the target. Then, to find this hypothesis with success probability $1-\delta/2$, we need another $O((1/\epsilon)\log(1/\delta))$ random example queries. 

To this end, consider the algorithm that takes a sample of size $m=O((1/\epsilon)s\log n)$ and finds a consistent hypothesis in $\DT^s_{2\log m}$ using Lemma~\ref{LearnDTds}. The probability that the algorithm fails to find a decision tree of size $s$ and depth $2\log m$ is at most the probability that one of the $m$ examples reaches depth $2\log m$ in the target decision tree. This probability is at most $m2^{-2\log m}=1/m\le 1/4$. Therefore, with probability at least $3/4$, the algorithm outputs a consistent hypothesis of size $s$, and by Lemma~\ref{occam} (with $H=\DT^s$) the result follows. By Lemma~\ref{LearnDTds}, such a decision tree can be constructed in time $n^{O(\log m)}m=m^{O(\log n)}.$ This implies the result. 
\end{proof}

The following theorem can be derived from Theorem~\ref{DTsU} and Lemma~\ref{Reduction}:

\begin{theorem}\label{DTsU2} The class $\DT^s$ is (uniform-distribution) properly learnable using $q=\tilde O((s/\epsilon)\log s$ $\log(1/\delta)))$  random example queries and $m=O(s\log n)$ black box queries in time $(s/\epsilon)^{O(\log s)}\log(1/\delta)+O((m+q)n)$.    
\end{theorem}

\subsection{Non-proper Learning}\label{NpL}
In this subsection, we  give results in non-proper learning of size-$s$ decision trees from larger classes of decision trees. 

\ignore{The following result is from~\cite{EhrenfeuchtH89}. Theorem~\ref{DTs} gives another simple proof
}
Theorem~\ref{DTs} is a result from~\cite{EhrenfeuchtH89}, with a simpler proof.

\begin{theorem}\label{DTs}{\rm \cite{EhrenfeuchtH89}}. The class $\DT^s$ is distribution-free learnable from $\DT^{n^{O(\log s)}}$ with $$q=\frac{1}{\epsilon}\left(n^{O(\log s)}+\log\frac{1}{\delta}\right)$$ random example queries in time $q n^{O(\log s)}$.  
\end{theorem}
\begin{proof}
Given a sample $S$ that is consistent with a function $f$ in $\DT^s$,
we first show how to construct a decision tree of size $n^{O(\log s)}$ that is consistent with $S$. We use the fact that either the right or the left sub-tree of the root of $f$ has size at most $s':=\lfloor s/2\rfloor$. Therefore, the algorithm first guesses the variable at the root, and then guesses which sub-tree has size at most $s'$. Then, it recursively constructs the tree of size $s'$. When it succeeds, it continues to construct the other sub-tree. If the number of guesses is $\phi(n,s)$, then we have
$\phi(n,s)\le (2n)\phi(n-1,\lfloor s/2\rfloor)+\phi(n-1,s)$. This gives $\phi(n,s)=n^{O(\log s)}$ guesses.
The size of the tree $S(n,s)$ satisfies $S(n,s)\le S(n-1,\lfloor s/2\rfloor)+S(n-1,s)$ which gives the bound $S(n,s)\le n^{O(\log s)}$. So the class of possible outputs $H$ has size at most $n^{n^{O(\log s)}}$. 

Hence, to learn the class, we use Lemma~\ref{occam}. The number of examples is 
$$q=\frac{1}{\epsilon}\left(\log |H| +\log \frac{1}{\delta}\right)= \frac{1}{\epsilon}\left(n^{O(\log s)} +\log \frac{1}{\delta}\right).$$ The time complexity is $q\phi(n,s)=q n^{O(\log s)}$.
\end{proof}

By Lemma~\ref{Reduction} and~\ref{DTs}, we have:
\begin{theorem}\label{DTsOUR} The class $\DT^s$ is distribution-free learnable from $\DT^{s^{O(\log s)}}$ in  $m=O(s\log n)$ black box queries and $q=(1/\epsilon)(s^{O(\log s)}+s\log(1/\delta))$ random example queries in time $q\cdot s^{O(\log s)}+O((q+m)n)$.   
\end{theorem}

\section{Testing Decision Tree}\label{TDT}
In this section, we give the results in the tables in Figure~\ref{TABLE3} and~\ref{TABLE4}.
\subsection{Reductions in Testing}\label{RiT}
In this subsection, we give some reductions that follow immediately from~\cite{Bshouty19b,Bshouty20}.

We say that a class $C$ is {\it closed under zero-one projection} if for every $f\in C$, every $i\in [n]$ and every $\xi\in\{0,1\}$ we have $f_{|x_i\gets \xi}\in C$. We say that $C$ is {\it symmetric} if for every permutation $\pi:[n]\to [n]$ and every $f\in C$ we have $f_\pi\in C$ where $f_\pi(x):=f(x_{\pi(1)},\cdots,x_{\pi(n)})$. All the classes in this paper are closed under zero-one projection and symmetric. 

The following results are proved in~\cite{Bshouty20} (Theorem 2) when $H=C$. The same proof works for any $C\subseteq H$. 

\begin{lemma}\label{ThTriv2v2}{\rm \cite{Bshouty20}}. Let $C\subseteq H$ be classes of Boolean functions (over $n$ variables) that are symmetric sub-classes of $k$-$\junta$ and are closed under zero-one projection. Suppose $C$ is distribution-free learnable from $H$ in time $T(n,\epsilon,\delta)$ using $Q(n,\epsilon,\delta)$ random example queries and $M(n,\epsilon,\delta)$ black-box queries. Then, there is a  distribution-free two-sided adaptive algorithm for $\epsilon$-testing $C$ by $H$ that makes $$m=\tilde O\left(M(k,\epsilon/12,1/24)+k\cdot Q(k,\epsilon/12,1/24)+\frac{k}{\epsilon}\right)$$ queries and runs in time $T(k,\epsilon/12,1/24)+O(mn)$.
\end{lemma}

\begin{lemma}\label{ThTriUv2}{\rm \cite{Bshouty20}}. Let $C\subseteq H$ be classes of Boolean functions (over $n$ variables) that are symmetric sub-classes of $k$-$\junta$ and are closed under zero-one projection. Suppose $C$ is learnable from $H$ under the uniform distribution in time $T(n,\epsilon,\delta)$ using $Q(n,\epsilon,\delta)$ random example queries and $M(n,\epsilon,\delta)$ black-box queries. Then, there is a (uniform distribution) two-sided adaptive algorithm for $\epsilon$-testing $C$ by $H$ that makes $$m=\tilde O\left(M(k,\epsilon/12,1/24)+ Q(k,\epsilon/12,1/24)+\frac{k}{\epsilon}\right)$$ queries and runs in time $T(k,\epsilon/12,1/24)+O(mn)$.
\end{lemma}

The following is proved in~\cite{Bshouty20} when $H=C$. The same proof gives the following result:
\begin{lemma}\label{ThTriv}{\rm \cite{Bshouty20}}. Let $C$ and $C\subseteq H$ be classes of Boolean functions that are symmetric sub-classes of $k$-$\junta$ and are closed under zero-one projection. Suppose there is a tester $\T$ for $C_k=\{f\in C| f$ is independent on $x_{k+1},\ldots,x_n\}$ such that
\begin{enumerate}
\item $\T$ is a  two-sided adaptive $\epsilon$-tester (resp. distribution-free $\epsilon$-tester) that runs in time $T(k,\epsilon,\delta)$.
\item If $f\in C_k$ then, with probability at least $1-\delta$, $\T$ accepts.
\item If $f$ is $\epsilon$-far from every function in $H_k$ (resp., with respect to $\D$) then, with probability at least $1-\delta$, $\T$ rejects.
\item $\T$ makes $Q(k,\epsilon,\delta)$ random example queries and $M(k,\epsilon,\delta)$ black-box queries.
\end{enumerate}
Then, there is a two-sided adaptive algorithm for $\epsilon$-testing (resp., distribution-free $\epsilon$-testing) $C$ by $H$ that makes $$m=\tilde O\left(M(k,\epsilon/12,1/24)+Q(k,\epsilon/12,1/24)+\frac{k}{\epsilon}\right)$$ 
(resp., makes $$m=\tilde O\left(M(k,\epsilon/12,1/24)+k\cdot Q(k,\epsilon/12,1/24)+\frac{k}{\epsilon}\right)$$ 
)
queries and runs in $T(k,\epsilon/12,1/24)+   O(mn)$ time. 
\end{lemma}

\subsection{Blanc-Lange-Tan Testers}\label{BLTT}
The following are the non-tolerant versions of Blanc et. al. testers~\cite{BlancLT20}.

\begin{lemma}\label{Blan1}{\rm \cite{BlancLT20}}.
There is a tester that runs in $poly(\log s,1/\epsilon)\tilde O(n)$ time, makes $poly(\log s,1/\epsilon)\cdot \log n$ queries to unknown function $f$, and
\begin{enumerate}
    \item Accepts w.h.p if $f$ is a size-$s$ decision tree.
    \item Rejects w.h.p if $f$ is $\epsilon$-far from size-$s^{\tilde O((\log s)^2/\epsilon^3)}$ decision trees.
\end{enumerate}
\end{lemma}

\begin{lemma}\label{Blan2}{\rm \cite{BlancLT20}}.
There is a tester that runs in $poly(d,1/\epsilon)\tilde O(n)$ time, makes $poly(d,1/\epsilon)\cdot \log n$ queries to unknown function $f$, and
\begin{enumerate}
    \item Accepts w.h.p if $f$ is a depth-$d$ decision tree.
    \item Rejects w.h.p if $f$ is $\epsilon$-far from size-$\tilde O(d^3/\epsilon^3)$ decision tree.
\end{enumerate}
\end{lemma}
\subsection{Distribution-Free Testers}
\label{DFT}
In this section, we prove the results in the table in Figure~\ref{TABLE3} for the distribution-free model.

\begin{theorem}\label{TH1}
There is a distribution-free tester that runs in $s^{O(d)}/\epsilon+n\cdot\tilde O(s^2/\epsilon)$ time, makes $\tilde O(s^2/\epsilon)$ queries to unknown function $f$, and
\begin{enumerate}
    \item Accepts w.h.p if $f$ is a depth-$d$ size-$s$ decision tree.
    \item Rejects w.h.p if $f$ is $\epsilon$-far from depth-$d$ size-$s$ decision trees.
\end{enumerate}
\end{theorem}
\begin{proof}
By Theorem~\ref{DTds}, a depth-$d$ size-$s$ decision tree can be properly learned with  $q=Q(n,\epsilon,\delta)=O((1/\epsilon)(s\log n+\log(1/\delta))$ random example queries in time $T(n,\epsilon,\delta)=qn^{O(d)}$. We have $C=\DT^s_d\subseteq s$-$\junta$ and therefore $k=s$. Then by Lemma~\ref{ThTriv2v2}, there is a  distribution-free two-sided adaptive algorithm for $\epsilon$-testing $\DT_d^s$ that makes $$m=\tilde O\left(s\cdot Q(s,\epsilon/12,1/24)+\frac{s}{\epsilon}\right)=\tilde O\left(\frac{s^2}{\epsilon}\right)$$ queries and runs in time $T(s,\epsilon/12,1/24)+O(mn)=s^{O(d)}/\epsilon +n\cdot\tilde O(s^2/\epsilon)$. This implies Theorem~\ref{TH1}.
\end{proof}

Using the Theorem~\ref{DTdsTPre} with Lemma~\ref{ThTriv2v2}, we get:

\begin{theorem}\label{TH1Uexp}
There is a distribution-free tester that runs in $\tilde O(2^{s})+\tilde O(s^2/\epsilon)n$ time, makes $\tilde O(s^2/\epsilon)$ queries to unknown function $f$, and
\begin{enumerate}
    \item Accepts w.h.p if $f$ is a depth-$d$ size-$s$ decision tree (resp. size-$s$ decision tree).
    \item Rejects w.h.p if $f$ is $\epsilon$-far from depth-$d$ size-$s$ decision trees (resp. size-$s$ decision tree).
\end{enumerate}
\end{theorem}

The following Theorem follows from Corollary~\ref{DT} and Lemma~\ref{ThTriv2v2}.

\begin{theorem}\label{TH12}
There is a distribution-free tester that runs in $2^{O(d^2)}+n\cdot\tilde O(2^{2d}+2^{d}/\epsilon)$ time, makes $\tilde O(2^{2d}+2^{d}/\epsilon)$ queries to unknown function $f$, and
\begin{enumerate}
    \item Accepts w.h.p if $f$ is a depth-$d$ decision tree.
    \item Rejects w.h.p if $f$ is $\epsilon$-far from depth-$d$ decision trees.
\end{enumerate}
\end{theorem}

\subsection{Uniform Distribution Testers}\label{UD2}
In this section, we prove the results in the table in Figure~\ref{TABLE3} for the uniform distribution model.

\begin{theorem}\label{TH1U}
There is a tester that runs in $s^{O(d)}/\epsilon+n\cdot\tilde O(s/\epsilon)$ time, makes $\tilde O(s/\epsilon)$ queries to unknown function $f$, and
\begin{enumerate}
    \item Accepts w.h.p if $f$ is a depth-$d$ size-$s$ decision tree.
    \item Rejects w.h.p if $f$ is $\epsilon$-far from depth-$d$ size-$s$ decision trees.
\end{enumerate}
\end{theorem}
\begin{proof}
By Theorem~\ref{DTds}, a depth-$d$ size-$s$ decision tree can be properly learned with  $q=Q(n,\epsilon,\delta)=O((1/\epsilon)(s\log n+\log(1/\delta))$ random example queries in time $T(n,\epsilon,\delta)=qn^{O(d)}$. We have $C=\DT^s_d\subseteq s$-$\junta$ and therefore $k=s$. Then by Lemma~\ref{ThTriUv2}, there is a  uniform distribution two-sided adaptive algorithm for $\epsilon$-testing $\DT_d^s$ that makes $$m=\tilde O\left(Q(s,\epsilon/12,1/24)+\frac{s}{\epsilon}\right)=\tilde O\left(\frac{s}{\epsilon}\right)$$ queries and runs in time $T(s,\epsilon/12,1/24)+O(mn)=s^{O(d)}/\epsilon +n\cdot\tilde O(s/\epsilon)$. This implies Theorem~\ref{TH1U}.
\end{proof}

The following is derived from Theorem~\ref{DTsU} and Lemma~\ref{ThTriUv2}.

\begin{theorem}\label{TH1UB}
There is a tester that runs in $(s/\epsilon)^{O(\log s)}+n\cdot \tilde O(s/\epsilon)$ time, makes $\tilde O(s/\epsilon)$ queries to unknown function $f$, and
\begin{enumerate}
    \item Accepts w.h.p if $f$ is a size-$s$ decision tree.
    \item Rejects w.h.p if $f$ is $\epsilon$-far from size-$s$ decision trees.
\end{enumerate}
\end{theorem}

Theorem~\ref{TH12U} follows from Theorem~\ref{TH1U} and the fact that $\DT_d=\DT_d^{2^d}.$

\begin{theorem}\label{TH12U}
There is a tester that runs in $2^{O(d^2)}/\epsilon+n\cdot\tilde O(2^{d}/\epsilon)$ time, makes $\tilde O(2^{d}/\epsilon)$ queries to unknown function $f$, and
\begin{enumerate}
    \item Accepts w.h.p if $f$ is a depth-$d$ decision tree.
    \item Rejects w.h.p if $f$ is $\epsilon$-far from depth-$d$ decision trees.
\end{enumerate}
\end{theorem}

\subsection{Testers by Larger Classes}\label{TbLC}
In this section, we prove the results in the table in Figure~\ref{TABLE4}.

Our first result follows from Theorem~\ref{DTsOUR} and Lemma~\ref{ThTriv2v2}. 
\begin{theorem}\label{TH1NP}
There is a distribution-free tester that runs in $s^{O(\log s)}n/\epsilon$ time, makes $\tilde O(s^{O(\log s)}/\epsilon)$ queries to unknown function $f$, and
\begin{enumerate}
    \item Accepts w.h.p if $f$ is a size-$s$ decision tree.
    \item Rejects w.h.p if $f$ is $\epsilon$-far from size-$s^{O(\log s)}$ decision trees.
\end{enumerate}
\end{theorem}

The following two results are concluded from Lemmas~\ref{ThTriv}, \ref{Blan1} and \ref{Blan2}.

\begin{theorem}\label{BlankTh1}
There is a (uniform distribution) tester that makes $q=poly(1/\epsilon,\log s)+\tilde O(s/\epsilon)$ queries to unknown function $f$, runs in $poly(1/\epsilon)\tilde O(s)+O(qn)$ time and
\begin{enumerate}
    \item Accepts w.h.p if $f$ is a size-$s$ decision tree.
    \item Rejects w.h.p if $f$ is $\epsilon$-far from size-$s^{O((\log^2 s)/\epsilon^3)}$ decision trees.
\end{enumerate}
\end{theorem}

\begin{theorem}\label{BlankTh2}
There is a (uniform distribution) tester that makes $q=poly(1/\epsilon,d)+\tilde O(2^d/\epsilon)$ queries to unknown function $f$, runs in $poly(1/\epsilon)\tilde O(2^d)+O(qn)$ time and
\begin{enumerate}
    \item Accepts w.h.p if $f$ is a depth-$d$ decision tree.
    \item Rejects w.h.p if $f$ is $\epsilon$-far from depth-${O(d^3/\epsilon^3)}$ decision trees.
\end{enumerate}
\end{theorem}

\section{A Tester for $\DT_d$}
In this section we prove the following result:
\begin{theorem}\label{TH33N}
There is a distribution-free tester that makes $q=\tilde O(2^{2d}/\epsilon)$ queries to unknown function $f$, runs in $O(qn)$ time and
\begin{enumerate}
    \item Accepts w.h.p if $f$ is a depth-$d$ decision tree.
    \item Rejects w.h.p if $f$ is $\epsilon$-far from depth-$d^2$ decision trees.
\end{enumerate}
\end{theorem}

\subsection{The Key Lemma}
We start with some notations and definitions, and then prove the key Lemma for the tester. 

Recall that {\it monomial} is a conjunction of variables. A $k$-monomial is a monomial with at most $k$ variables.
A polynomial (over the field $F_2$) is a sum (in the binary field $F_2$) of monomials. An $s$-sparse polynomial is a sum of at most $s$ monomials. We say that the polynomial $f$ is of degree-$d$ if its monomials are $d$-monomials. 

It is well known that (see for example Lemma~4 in~\cite{BshoutyH19}):
\begin{lemma}\label{BasDT}
A depth-$d$ decision tree is a $3^d$-sparse degree-$d$ polynomial with at most $2^d$ relevant variables.
\end{lemma}

Let $f=M_1+M_2+\cdots+M_t$ be a polynomial. We say that $M_i$ is a {\it maximal monomial} of $f$ if $M_i$ is not a sub-monomial of any other monomial $M_j$, i.e., for every other monomial $M_j$, there is a variable in $M_i$ that is not in $M_j$.

We now prove the key Lemma for our tester:
\begin{lemma}\label{KeyNew}
Let $f$ be a depth-$d$ decision tree and $f=M_1+M_2+\cdots+M_t$ be its polynomial representation. Let $M_i=x_{i_1}\cdots x_{i_{d'}}$, $d'\le d$ be a maximal monomial of $f$. For any $\xi_1,\xi_2,\ldots,\xi_{d'}\in\{0,1\}$, we have that $f_{|x_{i_1}\gets \xi_1,\cdots, x_{i_{d'}}\gets \xi_{d'} }$ is depth-$(d-1)$ decision tree.
\end{lemma}
\begin{proof}
The proof is by induction on the number of variables $m$ of $f$. The case $m=1$ is trivial. Now assume that the result holds for any $m\le k$.

Let $T$ be any depth-$d$ decision tree with $k+1$ variables that represents $f$. Let $M_i=x_{i_1}\cdots x_{i_{d'}}$, $d'\le d$ be a maximal monomial of $f$. Let $X=\{x_{i_1},\ldots,x_{i_{d'}}\}$. If the variable of the root of $T$ is $x_{i_j}\in X$ then, $T_{|x_{i_j}\gets 0}$ and $T_{|x_{i_j}\gets 1}$ are left and right decision sub-trees of $T$ and are of depth at most $d-1$. Then, $T_{|x_{i_1}\gets \xi_1,\cdots, x_{i_{d'}}\gets \xi_{d'}}$ is of depth at most $d-1$ for any $\xi_1,\ldots,\xi_{d'}\in\{0,1\}$.

If the variable of the root of the tree $T$ is $x_\ell\not\in X$, then the left sub-tree $T_{|x_\ell\gets 0}$ is a depth-$(d-1)$ decision tree and has at most $k$ variables. We now claim that $M_i$ is a maximal monomial of $T_{|x_\ell\gets 0}$. This is because of the fact that substituting $x_\ell=0$ in the polynomial representation only removes monomials in $f$. Since $x_\ell$ is not in $M_i$, it does not remove $M_i$. Therefore, $M_i$ is maximal monomial in $T_{|x_\ell\gets 0}$, and by the induction hypothesis $T_{|x_\ell\gets 0,x_{i_1}\gets \xi_1,\cdots, x_{i_{d'}}\gets \xi_{d'}}$ is depth-$(d-2)$ decision tree. 

The right sub-tree $T_{|x_\ell\gets 1}$ is a depth-$(d-1)$ decision tree that has at most $k$ variables. We now claim that $T_{|x_\ell\gets 1}$ also has $M_i$ as a monomial and it is maximal. For if it is removed or not maximal, then there must be the monomial $M_j=x_{\ell}M_i$ in $T$. Since $M_i$ is sub-monomial of $M_j$ we get a contradiction. Therefore, by the induction hypothesis $T_{|x_\ell\gets 1,x_{i_1}\gets \xi_1,\cdots, x_{i_{d'}}\gets \xi_{d'}}$ is depth-$(d-2)$ decision tree. 

This implies that $$T_{|x_{i_1}\gets \xi_1,\cdots, x_{i_{d'}}\gets \xi_{d'}}=x_\ell\cdot T_{|x_\ell\gets 1,x_{i_1}\gets \xi_1,\cdots, x_{i_{d'}}\gets \xi_{d'}}+\overline{x_\ell}\cdot T_{|x_\ell\gets 0,x_{i_1}\gets \xi_1,\cdots, x_{i_{d'}}\gets \xi_{d'}}$$ is depth-$(d-1)$ decision tree. 
\end{proof}

For every Boolean function $f$ that is $3^d$-sparse degree-$d$ polynomial, we define the following decision tree $T_f$. If $f$ is constant function, then $T_f$ is a leaf labeled with this constant. Let $f=M_1+M_2+\cdots+M_t$. Consider any maximal monomial $M_i$ of $f$. Let $M_i=x_{i_1}\cdots x_{i_{d'}}$ where $i_1<i_2<\cdots<i_{d'}$. The tree $T_f$ has all the variables $x_{i_1},\cdots, x_{i_{d'}}$ at the first $d'$ levels of the tree. That is, the first $d'$ levels of the tree is a complete tree where the label of all the nodes at level $j$ is $x_{i_j}$. So every $\xi_1,\ldots,\xi_{d'}\in\{0,1\}$ leads to a different vertex at level $d'$ in $T_f$ from which we recursively attach the decision tree $T_{|g}$ where $g=f_{|x_{i_1}\gets \xi_1,\cdots, x_{i_{d'}}\gets \xi_{d'}}$.

We now prove:
\begin{lemma} Let $f$ be a sparse-$3^d$ degree-$d$ polynomial. Then, for $h=d(d-1)/2$:
\begin{enumerate}
    \item\label{DTN1} If $f$ is a depth-$d$ decision tree, then $T_f$ is depth $h$-decision tree.
    \item\label{DTN3} If $f$ is $\epsilon$-far from every depth $h+1$ decision tree according to the distribution $\D$ then, for a random assignment $a$ drawn according to the distribution $\D$, with probability at least $\epsilon$, the path that $a$ takes in $T_{f}$ reaches depth $h+1$.
\end{enumerate}
\end{lemma}
\begin{proof}
{\it \ref{DTN1}} follows immediately from Lemma~\ref{KeyNew}.

To prove {\it \ref{DTN3}}, let $T'_f$ be the tree $T_f$ where every vertex of depth $h+1$ is changed to a leaf labeled with $0$. Since $f$ is $\epsilon$-far from every depth $h+1$ decision tree according to the distribution $\D$, it is $\epsilon$ far from $T'_f$. Since the leaves of $T'_f$ of depth at most $h$ correctly compute $f$, the probability that a random assignment $a$ chosen according to distribution $\D$ ends up in a leaf of depth at most $h$ is less than $1-\epsilon$. This completes the proof.
\end{proof}

\subsection{The Tester}
In this section, we prove Theorem~\ref{TH33N}. To that end, we start by the following Lemma:

\ignore{In the next subsection~\ref{RVMM}, we prove:}
\begin{lemma} \label{TimeV}We have
\begin{enumerate}
    \item There is an algorithm that for a degree-$d$ polynomial $f$ makes $q=\tilde O(2^{2d})+2^d\log n$ queries, runs in time $O(qn)$ and finds the relevant variables of $f$.
    \item There is an algorithm that for a degree-$d$ polynomial over $2^d$ variables $X$ makes $q'=\tilde O(2^d)$ queries, runs in time $O(q'n)$ and finds a maximal monomial in $f$.
\end{enumerate}
\end{lemma}
The proof of Lemma \ref{TimeV} is given in subsection~\ref{RVMM}. This immediately gives the following result that we need for our tester.
\begin{lemma} Let $f$ be a sparse-$3^d$ degree-$d$ polynomial over $2^d$ variables. Let $h=d(d-1)/2$.
 Given the relevant variables of $f$, for any $a\in \{0,1\}^n$, the path that $a$ takes in $T_{f}$ up to depth at most $h+1$ can be computed in $\tilde O(2^{d})h$ time.
\end{lemma}

The tester's paradigm is as follows. First, the tester finds the relevant variables of $f$ (See details in the next subsection). If the number of relevant variables exceeds $2^d$, then the tester rejects. The tester then, for $t=O(1/\epsilon)$ assignments $a^{(1)},\ldots,a^{(t)}$ drawn according to the distribution $\D$, finds the route of each $a^{(i)}$ in $T_f$. If no maximal monomial of size at most $d$ can be found then the tester rejects. 
If one of the routes exceeds depth $h=d(d-1)/2$, the algorithm rejects. Otherwise it accepts. Each route takes time $\tilde O(2^d)$. So the number of queries is $q=\tilde O(2^{2d})/\epsilon+2^d\log n$ and the time is $O(qn)$. By Lemma~\ref{ThTriv} the result follows.

In the Appendix, we give another algorithm that has the same query complexity for $\epsilon>2^{-d^2}$. 

\subsection{Relevant Variables and Maximal Monomial}\label{RVMM}
In this subsection, we prove Lemma~\ref{TimeV}. We show how to find the relevant variables and a maximal monomial of any degree-$d$ polynomial. 

The following is a very well known result:
\begin{lemma}\label{BasE}
For any non-constant degree-$d$ polynomial $f$, we have $\Pr[f(x)\not=f(0)]\ge 1/2^d$.
\end{lemma}

The following is a well known result in learning theory. We prove it for completeness.
\begin{lemma}~\label{Rell}
There is an algorithm that given any degree-$d$ polynomial $f$ over $v$ variables and a set $X$ of some of its relevant variables, asks $2^d\log(1/\delta)+\log v$ queries and, with probability at least $1-\delta$, decides if $X$ are all its relevant variables, and if not, finds a new relevant variable of $f$.  
\end{lemma}
\begin{proof}
Let $X'=\{x_{i_1},\ldots,x_{i_t}\}$ be the set of variables that are not in $X$. Define $g=f+f_{|X'\gets 0}$. Since $g$ is of degree at most $d$, by Lemma~\ref{BasE}, with $2^d\log(1/\delta)$ queries to $g$, with probability at least $1-\delta$, we can decide if $g$ is a constant function. If not, we get an assignment $a$ such that $g(a)\not=g(0)$. If $g$ is constant function $\tau\in\{0,1\}$, then $f(x)=f_{|X'\gets 0}+\tau$ and $f$ is independent of $X'$. So, the variables in $X$ are all the relevant variables of $f$

If $g(a)\not=g(0)$, then since $g(0)=0$, we have $f(a)\not= f_{|X'\gets 0}(a)=f(a_{|X'\gets 0})$. Now recursively flip half of the entries of $a$ that differ from $a_{|X'\gets 0}$ and ask a query and keep the two assignments that have different values in $f$. Eventually, we get an entry $a_{i_{k+1}}$ that flipping it changes the value of the function. Then, $x_{i_{k+1}}$ is relevant variable in $f$. Now $x_{i_{k+1}}\not\in X$ is a new relevant variable because the entries of each $x_{i_j}\in X$ in $a$ agree with the value of the same entry in $a_{|X'\gets 0}$. 
The number of queries in this procedure is at most $\log v$. This completes the proof.
\end{proof}

Therefore, for degree-$d$ polynomial, choosing confidence $\delta/2^d$ in the above algorithm, with probability at least $1-\delta$, we find the first $2^d$ relevant variables of $f$ using $\tilde O(2^{2d})\log(1/\delta)+2^d\log n$ queries. If $f$ has more than $2^d$ relevant variables, then $f$ is not a depth-$d$ decision tree and the tester rejects.

We now prove

\begin{lemma} Let $f$ be a degree-$d$ polynomial and $X$ be the set of its relevant variables.
Let $M=x_{i_1}\cdots x_{i_k}$ be a sub-monomial of some monomial of $f$ ($M$ is not necessarily a monomial of $f$). There is an algorithm that asks $O(2^d\log(1/\delta)+2^k\log |X|)$ queries and, with probability at least $1-\delta$, decides if $M$ is a maximal monomial of $f$, and if it is not, it finds a new variable $x_{i_{k+1}}$ such that $M'=x_{i_1}\cdots x_{i_k}x_{i_{k+1}}$ is a sub-monomial of some monomial of $f$.
\end{lemma}
\begin{proof}
Define the function 
$$G(x)=1+\sum_{(\xi_1,\ldots,\xi_k)\in\{0,1\}^k}f_{|x_{i_1}\gets \xi_1,\cdots, x_{i_k}\gets \xi_k}(x).$$
We prove the following:
\begin{enumerate}
\item A query to $G$ can be simulated by $2^k$ queries to $f$.
\item $G$ is a polynomial of degree at most $d-k$.
\item $M$ is maximal monomial of $f$ if and only if $G=0$.
\item If $G\not= 0$, then for any relevant variable $x_{i_{k+1}}$ of $G$, $M'=x_{i_1}\cdots x_{i_k}x_{i_{k+1}}$ is a sub-monomial of some monomial of $f$.
\end{enumerate}
The first item is obvious. We prove 2-4. Since $M$ is a sub-monomial of some monomial of $f$, we have $f=Mg+h$, where $g$ is a polynomial of degree at most $d-k$ (independent of $x_{i_1},\ldots,x_{i_k}$) and $h$ is a polynomial of degree at most $d$ that $M$ is not sub-monomial of any of its monomials. Notice that $M$ is maximal monomial of $f$ if and only if $g=1$. For a monomial $M''$ in $h$, we have that some variable in $M$, say w.l.o.g. $x_{i_1}$, is not in $M''$ and therefore,
$$\sum_{(\xi_1,\ldots,\xi_k)\in\{0,1\}^k}M''_{|x_{i_1}\gets \xi_1,\cdots, x_{i_k}\gets \xi_k}(x)=\sum_{\xi_1\in\{0,1\}}\sum_{(\xi_2,\ldots,\xi_k)\in\{0,1\}^{k-1}}M''_{|x_{i_2}\gets \xi_2,\cdots, x_{i_k}\gets \xi_k}(x)=0.$$

Then,
\begin{eqnarray*}
G(x)+1&=&\sum_{(\xi_1,\ldots,\xi_k)\in\{0,1\}^k}f_{|x_{i_1}\gets \xi_1,\cdots, x_{i_k}\gets \xi_k}(x)\\
&=&g(x)\sum_{(\xi_1,\ldots,\xi_k)\in\{0,1\}^k}M_{|x_{i_1}\gets \xi_1,\cdots, x_{i_k}\gets \xi_k}(x)+\sum_{(\xi_1,\ldots,\xi_k)\in\{0,1\}^k}h_{|x_{i_1}\gets \xi_1,\cdots, x_{i_k}\gets \xi_k}(x)\\
&=&g(x)+\sum_{(M''\mbox{\ {\small monomial in}}\ h)}\sum_{(\xi_1,\ldots,\xi_k)\in\{0,1\}^k}M''_{|x_{i_1}\gets \xi_1,\cdots, x_{i_k}\gets \xi_k}(x)\\
&=&g(x).
\end{eqnarray*}
Hence,$f=MG+M+h$ and the results 2-4 follows.

By Lemma~\ref{Rell}, there is an algorithm that asks $2^{d-k}\log(1/\delta)+\log|X|$  queries to $G$ (and therefore, $O(2^d\log(1/\delta)+2^k\log |X|)$ queries to $f$) and, with probability at least $1-\delta$, either decides that $G=0$, in which case $M$ is maximal monomial, or finds a new relevant variable $x_{i_{k+1}}$ of $G$, in which case $M'=x_{i_1}\cdots x_{i_k}x_{i_{k+1}}$ is a sub-monomial of some monomial of $f$. 
\end{proof}

For degree-$d$ polynomials with $|X|\le 2^d$ relevant variables, we choose confidence $\delta/2^d$ in the above algorithm and then, with probability at least $1-\delta$, we find a maximal monomial of $f$. 

\section{A Tester for $\DT^s$}\label{ATfDTs}
In this section we prove:
\begin{theorem}\label{TH44}
There is a (uniform-distribution) tester that makes $q=poly(s,1/\epsilon)$ queries to unknown function $f$, runs in $poly(s,1/\epsilon)n$ time and
\begin{enumerate}
    \item Accepts w.h.p if $f$ is a size-$s$ decision tree.
    \item Rejects w.h.p if $f$ is $\epsilon$-far from size-$(s/\epsilon)^{O(\log (s/\epsilon))}$ decision trees.
\end{enumerate}
\end{theorem}

\subsection{Preliminary Results}
 For a Boolean function $f$, the {\it constant-depth of $f$}, $cd(f)$, is the minimum number $\ell$ of variables $X=\{x_{j_1},\cdots,x_{j_\ell}\}$ such that $f_{|X\gets 0}$ is a constant function. We define $\M(f)$ the set of all monomials in the minimum size polynomial representation of $f$. Since minimum size polynomial representation of a Boolean function is unique, $\M(f)$ is well defined. For a set of monomials $S$, we denote $\Sigma S=\sum_{M\in S}M$. Notice that $\Sigma\M(f)=f$ and $\M(\Sigma S)=S$.
For any Boolean function and an interval $I$ (such as $[d]=\{1,2,\ldots,d\}, [d_1,d_2]=[d_2]\backslash [d_1-1]$ or $(d_1,d_2]=[d_2]\backslash [d_1]$), we define $f^I=\sum_{M\in \M(f) \mbox{\ {\small and}\ } |M|\in I}M$ where $|M|$ is the number of variables in $M$. 

We first prove:
\begin{lemma}\label{fIlf}
For any $S'\subseteq \M(f)$, we have $cd(\Sigma {S'})\le cd(f)$. In particular, for any interval $I$, we have $cd(f^I)\le cd(f)$.
\end{lemma}
\begin{proof}
Let $cd(f)=r$. Then, there is a set $X=\{x_{j_1},\cdots,x_{j_\ell}\}$ such that $f_{|X\gets 0}$ is constant. Therefore, for every non-constant $M\in \M(f)$, we have $M_{|X\gets 0}=0$. Then, $(\Sigma {S'})_{|X\gets 0}$ is constant and $cd(\Sigma {S'})\le cd(f)$.
\end{proof}

\begin{lemma}\label{Freq}
Let $f$ be any Boolean function. If $cd(f)\le \ell$, then there is a variable $x_i$ that appears in at least $1/\ell$ fraction of the non-constant monomials of $f$.
\end{lemma}
\begin{proof} If $cd(f)\le \ell$, then there is a set
$X=\{x_{j_1},\cdots,x_{j_\ell}\}$ such that $f_{|X\gets 0}$ is a constant function. This implies that for every non-constant monomial, there is $x_i\in X$ that appears in it. By the Pigeon hole principle the result follows. 
\end{proof}

Let $a\in\{0,1\}^n$ be a random uniform assignment. Consider, $$T(x)=f(x+a).$$ 

We now prove:
\begin{lemma}\label{subD} Let $f$ be a size-$s$ decision tree and let $T(x)=f(x+a)$ for a random uniform assignment $a\in\{0,1\}^n$. For a random uniform $\xi_1,\ldots,\xi_j\in\{0,1\}$
and any variables $x_{i_1},\ldots,x_{i_j}$ where each $i_\ell$ may depend on $T$, $a$, $i_1,\ldots,i_{\ell-1}$ and $\xi_1,\ldots,\xi_{\ell-1}$ but independent of $\xi_{\ell},\ldots,\xi_{j}$ and $q=(x_{i_1}\gets \xi_1,\ldots,x_{i_j}\gets \xi_j)$, with probability at least $1-s2^{-h}$,  $$cd\left(T_{|q}\right)\le h.$$
\end{lemma} 
\begin{proof}
We have $T_{|x_{i_1}\gets \xi_1,\ldots,x_{i_j}\gets\xi_j}(0)=T(0_{|x_{i_1}\gets \xi_1,\ldots,x_{i_j}\gets\xi_j})=f(a+0_{|x_{i_1}\gets \xi_1,\ldots,x_{i_j}\gets\xi_j})$. Since $b:=a+0_{|x_{i_1}\gets \xi_1,\ldots,x_{i_j}\gets\xi_j}$ is random uniform in $\{0,1\}^n$, the path that $b$ takes in the computation of $f(b)$ is a random uniform path in $f$. With probability at least $1-s2^{-h}$, this path reaches depth less than or equal to $h$ in $f$. Therefore, with probability at least $1-s2^{-h}$, there are $h'\le h$ variables $x_{j_1},\ldots,x_{j_{h'}}$ (the variables in this path) such that $f_{x_{j_1}\gets b_{j_1},\ldots,x_{j_{h'}}\gets b_{j_{h'}}}$ is constant, say $\tau\in\{0,1\}$. Let $J=\{j_1,j_2,\ldots,j_{h'}\}$ and $I=\{i_1,i_2,\ldots,i_j\}$ and suppose, w.l.o.g, $J\cap I = \{i_1,\ldots,i_r\}$. Then,
\begin{eqnarray*}
\tau&=&f_{|x_{j_1}\gets b_{j_1},\ldots,x_{j_{h'}}\gets b_{j_{h'}}}(x)\\
&=&
f_{|x_{j_1}\gets b_{j_1},\ldots,x_{j_{h'}}\gets b_{j_{h'}}}(x+a)\\
&=& T_{|x_{j_1}\gets b_{j_1}+a_{j_1},\ldots,x_{j_{h'}}\gets b_{j_{h'}}+a_{j_{h'}}}\\
&=& \left(T_{|x_{i_1}\gets \xi_1,\ldots,x_{i_r}\gets \xi_{r}}\right)_{|(J\backslash I)\gets 0}\\
\end{eqnarray*}
and thus,
$$(T_{|q})_{|(J\backslash I)\gets 0}=(T_{|x_{i_1}\gets \xi_1,\ldots,x_{i_j}\gets\xi_j})_{|(J\backslash I)\gets 0}=(T_{|x_{i_1}\gets \xi_1,\ldots,x_{i_r}\gets\xi_r})_{|(J\backslash I)\gets 0,x_{i_{r+1}}\gets \xi_{r+1},\ldots,x_{i_{j}}\gets \xi_{j}}=\tau.$$ Therefore, 
$$cd\left(T_{|q}\right)\le |J\backslash I|\le h'\le h.$$
\end{proof}
Let
\begin{eqnarray} \label{The-r}
r=\log(s/\epsilon).
\end{eqnarray}

The tester will use Lemma~\ref{subD}, $poly(\log (s/\epsilon))/\epsilon$ times, therefore we can choose $h=2r$ which, by union bound, adds a failure probability of $(poly(\log (s/\epsilon))/\epsilon)\cdot s2^{-h}=\tilde O(\epsilon/s)$ to the tester. Let $E_1$ be the event  
\begin{eqnarray}\label{Reslog}
E_1:\ \ (\forall q)\ cd(T_{|q})\le 2r
\end{eqnarray}
for all the $q$ that are generated in the tester.

For the rest of this section, we let $f$ be a size-$s$ decision tree. Let $f=f_1+f_2+\cdots+f_s$ be the disjoint-terms sum representation of $f$.\footnote{Every size-$s$ decision tree can be represented as a sum of terms $T_1+T_2+\cdots+T_{s'}$, $s'\le s$, where $T_i\wedge T_j=0$ for every $i\not=j$. The number of terms $s'$ is the number of leaves labeled with $1$. See for example~\cite{BshoutyM02}. Here we assume $s'=s$ because we can always change a term $t$ to $tx_j+t\overline{x_j}$.} Let  $T_i=f_i(x+a)$ for $i\in [s]$. It is easy to see that $T=T_1+T_2+\cdots+T_s$ is disjoint-terms sum representation of $T$. We denote by $T_i^+$ the conjunction of the non-negated variables in $T_i$.

We now prove:
\begin{lemma}\label{DTs-base} Let $\lambda$ be any constant. For a random uniform $a$, with probability at least $1-s(\epsilon/s)^\lambda$ the following event $E_2(\lambda)$ holds
\begin{eqnarray}\label{LarGe}
E_2(\lambda):\ \ \mbox{For every $i$, if $|T_i|>16\lambda r$ then  $|T_i^+|\ge 4\lambda r$.}
\end{eqnarray}
\end{lemma}
\begin{proof}
Since $T_i(x)=f_i(x+a)$ and $a$ is random uniform, each variable in $T_i$ is positive with probability $1/2$. By Chernoff bound the result follows.
\end{proof}

To change the disjoint-terms representation of $T$ to polynomial representation, we take every term $T_i$ and expand it to sum of monomials\footnote{For example $x_1x_2\bar x_3\bar x_4=x_1x_2(x_3+1)(x_4+1)=x_1x_2x_3x_4+x_1x_2x_3+x_1x_2x_4+x_1x_2$.}. A monomial that is generated from even number of different terms will not appear in the polynomial, while, those that are generated from odd number of different terms will appear in the polynomial. 

Let $M_1+M_2+\cdots+M_{\ell}$ be the multivariate polynomial representation of $T$, where $|M_1|\le |M_2|\le \cdots \le |M_{\ell}|$. Note that $\ell$ can be exponential in~$n$. We say that $M_i$ {\it is generated by} $T_j$ if $j$ is the smallest integer for which $T_j$ generates $M_i$. The following is a trivial result:
\begin{lemma}
\label{ksl} If $M_i$ is generated by $T_j$, then $|T_j|\ge |M_i|\ge |T_j^+|$ and $T_j^+$ is a sub-monomial of $M_i$. That is, $M_i=T_j^+M_i'$ for some monomial $M_i'$.
\end{lemma}

Using Lemma~\ref{DTs-base}, we prove:
\begin{lemma}\label{NuMb}
If $E_2(\lambda)$ in (\ref{LarGe}), holds then the number of monomials $M_i$ of size at most $4\lambda r$ is at most $s(s/\epsilon)^{16\lambda}.$
\end{lemma}
\begin{proof}
By Lemma~\ref{DTs-base}, the monomials that has size at most $4\lambda r$ are generated from terms of size at most $16\lambda r$. We have at most $s$ terms of size at most $16\lambda r$ and each one generates at most $2^{16\lambda r}$ monomials. So we have at most $s(s/\epsilon)^{16\lambda }$ such monomials. 
\end{proof}

\begin{lemma}\label{En}
If $E_2(\lambda)$ in (\ref{LarGe}) holds, then there are $s$ monomials $N_1,N_2,\ldots, N_s$, each of size at least $4\lambda r$, such that for every monomial $M_i$ (of $T$) of size at least $16\lambda r$, there is $N_j$ where $M_i=N_jM_i'$ for some monomial $M_i'$. 
\end{lemma}
\begin{proof}
Define $N_i=T_i^+$ if $|T_i^+|\ge 4\lambda r$n and $N_i=x_1x_2\cdots x_n$ otherwise. Let $M_i$ be a monomial of $T$ of size at least $16\lambda r$. 
By Lemma~\ref{ksl}, $M_i$ is generated by a monomial $T_j$ of size at least $16\lambda r$. By Lemma~\ref{DTs-base}, $|T_j^+|\ge 4\lambda r$.  By Lemma~\ref{ksl}, $N_j=T_j^+$ is a sub-monomial of $M_i$. 
\end{proof}
\ignore{Consider $T^{[16\lambda r]}$ and $T^{(16\lambda r,s]}$.  By the above Lemma we can write $T$ as
\begin{eqnarray}\label{RepT}
T=T^{[16\lambda r]}+\sum_{i=1}^sN_iT^{(i)}
\end{eqnarray}
for some multivariate polynomials $T^{(i)}$ and monomials $N_i$ of size at least $4\lambda r$. 

In particular, 
\begin{lemma}\label{AppR}
We have
$$\Pr\left[T^{(16\lambda r,s]}=1\right]=\Pr\left[T\not=T^{[16\lambda r]}\right]\le s\left(\frac{\epsilon}{s}\right)^{4\lambda }.$$
Also, 
$$\Pr\left[\vee \M(T^{(16\lambda r,s]})=1\right]\le s\left(\frac{\epsilon}{s}\right)^{4\lambda }.$$
\end{lemma}}

We remind the reader that for an interval $R$, $T^R$ is the sum of the monomials $M$ of $T$ of size $|M|\in R$. For a set of monomials $A$, we denote $\vee A=\vee_{M\in A}M$. We will also need the following lemma:
\begin{lemma}\label{smA}
Suppose $E_2(\lambda)$ holds. Let $P= \vee \M(T^{(16\lambda r,s]})$. For a random uniform $\xi_1,\ldots,\xi_j\in\{0,1\}$
and any variables $x_{i_1},\ldots,x_{i_j}$ where each $i_\ell$ may depend on $T$, $i_1,\ldots,i_{\ell-1}$ and $\xi_1,\ldots,\xi_{\ell-1}$ but independent of $\xi_{\ell},\ldots,\xi_{j}$ and $q=(x_{i_1}\gets \xi_1,\ldots,x_{i_j}\gets \xi_j)$, with probability at least $1-s(\epsilon/s)^{2\lambda}$,
$$\Pr\left[ P_{|q}=1\right]\le s\left(\frac{\epsilon}{s}\right)^{2\lambda }.$$ 
\end{lemma}
\begin{proof} By Lemma~\ref{En}, we have:
$$P=N_1P_1\vee N_2P_2\vee\cdots\vee N_sP_s,$$ where $P_i$ is a Boolean function and $N_i$ is a monomial of size at least $4\lambda r$ for each $i\in\{1,\dots, s\}$. The probability that each $(N_i)_{|q}$ is not zero and is of size at most $2\lambda r$ is at most $2^{-2\lambda r}=(\epsilon/s)^{2\lambda}$. Since,
for all $i\in [s]$, $(N_i)_{|q}$ is zero or $|(N_i)_{|q}|>2\lambda r$ implies 
$\Pr[P_{|q}=1]\le s(\epsilon/s)^{2\lambda}$, the result follows. 
\end{proof}

\subsection{The Tester}
In this subsection, we give the tester and prove its correctness. Recall that $$r=\log(s/\epsilon).$$ Let $c\ge 2$ be any constant. The tester first chooses a random uniform $a\in\{0,1\}^n$ and defines $T(x)=f(x+a)$. Then, it learns all the monomials of size $16r'$ where\footnote{Here $c$ can be $2$. We kept it to show the effect of this constant on the success probability of the tester.}
$$r'=16cr.$$ This can done by the algorithm in~\cite{BshoutyM02} in $poly(n,s/\epsilon)$ time and queries. We will show later how to eliminate $n$ in the query complexity. 

Then, the tester splits the monomials of size at most $16r'$ to monomials of size less or equal to $r'$ (the function $F$) and those that have size between $r'$ and $16r'$ (the function $G$). The tester performs $O(1/\epsilon)$ random walks in a decision tree that at each stage, $j$, deterministically chooses a variable $x_{i_j}$ with the smallest index $i_j$ that appears in at least $1/(2r)$ fraction of the monomials of $\Sigma {H_j}$ and  chooses a random $\xi_j\in\{0,1\}$ for $x_{i_j}$. The set $H_j$ is defined in a way that (1) it is a subset of the monomials of the function $F_{|x_{i_1}\gets \xi_1,\ldots,x_{i_{j-1}\gets \xi_{j-1}}}$ and (2) it contains a variable that appears in at least $1/(2r)$ fraction of the monomials. See the details below. 

Although this decision tree is not the decision tree of $F$, we can still show that when $f$ is size-$s$ decision tree, with probability at least $2/3$ the random walk ends after $O(\log^2(s/\epsilon))$ steps and then the tester accepts. When it is $\epsilon$-far from any size-$(s/\epsilon)^{O(\log(s/\epsilon))}$ decision tree then, with probability at least $2/3$, something goes wrong (the learning algorithm fails or no variable appears in at least $1/(2r)$ fraction of the monomials of $H_j$) or the random walk does not end after $\Omega(\log^2(s/\epsilon))$ steps and then it rejects.

The tester query complexity is $poly(n,s/\epsilon)$. We use the reduction from \cite{Bshouty20} to change the query complexity to $poly(s/\epsilon)$.
 
\mybox{ {\bf Test$(f)$}\\ \hspace{.3in}
{\bf Input}: black box access to $f$\\
{\bf Output}: Accept or Reject
\begin{enumerate}
    \item $T\gets f(x+a)$ for random uniform $a\in\{0,1\}^n$. 
    \item Learn $T^{[16r']}$. If FAIL then Reject.
    \item $F\gets T^{[r']}$; $G\gets T^{(r',16r']}$; $H_0\gets \M(F)$; $L_0\gets \M(G)$.
    \item Repeat $40/\epsilon$ times
    \item \hspace{.15in} $j\gets 0$; $q_0\gets$ Empty sequence.
    \item \hspace{.15in} While $\Sigma {H_j}$ is not constant and $j<2^{10}c(\log^2(s/\epsilon))$  do 
    \item \hspace{.4in} $j\gets j+1$.
    \item \hspace{.4in} Find a variable $x_{i_j}$ with the smallest index that   
    \item[\ ] \hspace{1in} appears in at least $1/(2r)$ fractions of the monomials in $H_{j-1}$.
    \item \hspace{1in} If no such variable exists then Reject.
    \item \hspace{.4in} Choose a random uniform $\xi_j\in\{0,1\}$.
    \item \hspace{.4in} $q_j\gets (q_{j-1}; x_{i_j}\gets \xi_j).$\ \ 
    
    \hspace{1in} I.e., add to the list of  substitutions $q_{j-1}$ the substitution $x_{i_j}\gets \xi_j$.
    \item\label{kk1} \hspace{.4in} $H_j\gets \M((\Sigma {H_{j-1}})_{|x_{i_j}\gets \xi_j})\backslash \M((\Sigma {L_{j-1}})_{|x_{i_j}\gets \xi_j})$;
    \item\label{kk2} \hspace{.4in} $L_j\gets \M((\Sigma {L_{j-1}})_{|x_{i_j}\gets \xi_j})\backslash \M((\Sigma {H_{j-1}})_{|x_{i_j}\gets \xi_j})$;
    \item \hspace{.15in} If $j=2^{10}c(\log^2(s/\epsilon))$ then Reject.
    \item\label{LlL} If $\Pr[T_{|q_j}=1]$ is in $[\epsilon/4,1-\epsilon/4]$ then Reject
    \item Accept.
\end{enumerate}}
\\

\begin{lemma}\label{Redx} Assume that the event $E_2(16c)$ in (\ref{LarGe}) holds.
Let $F=T^{[r']}$ and $G=T^{(r',16r']}$. Let $x_{i_1},\ldots,x_{i_j}$ be variables, $\xi_1,\ldots,\xi_j$ be random uniform values in $\{0,1\}$, $q_j=(x_{i_1}\gets \xi_1,\ldots,x_{i_j}\gets \xi_j)$, $q_{j+1}=(q_j,x_{i_{j+1}}\gets \xi_{j+1})$, $H_j$, $H_{j+1}$ and $L_j$ as defined in the algorithm {\bf Test}$(f)$. Then, with probability at least $1-s(\epsilon/s)^{3c}$, we have
\begin{enumerate}
    \item\label{Mone} \begin{eqnarray}\label{Mai}
 H_j\subseteq S\left((T_{|q_j})^{[r']}\right).
\end{eqnarray}
\end{enumerate}

Let $E$ be the event that (\ref{Mai}) holds for all $j\le 2^{10}c\log^2(s/\epsilon)$ and all the $40/\epsilon$ random walks of the tester. Then,
$$\Pr[E]\ge 1-(\epsilon/s)^{2c}.$$

Assuming that $E$ holds, then,
\begin{enumerate}\setcounter{enumi}{1}
\item\label{Mtwo} There is a variable $x_{i_{j+1}}$ that appears in $1/(2r)$ fraction of the monomials in $H_j$.
\item\label{Mthree} If $\xi_{j+1}=0$, then $$|H_{j+1}|\le \left(1-\frac{1}{2r}\right) |H_{j}|.$$

\item\label{Mfour} If $\xi_{j+1}=1$, then $$|H_{j+1}|\le  |H_{j}|.$$
\end{enumerate} 

\end{lemma} 
\begin{proof}
Let $T=F+G+W$ such that $W=T^{(16r',s]}$. We first show that with probability at least $1-s(\epsilon/s)^{3c}$, we have $(W_{|q_j})^{[r']}=0$. Consider $N_1,N_2,\ldots,N_s$ in Lemma~\ref{En}. Every monomial in $W$ is of the form $MN_i$ for some monomial $M$ and $i\in [s]$. We also have $|N_i|\ge 4r'$ for all $i\in [s]$. The probability that for some $i\in[s]$ we have that $(N_i)_{|q_j}$ is not zero and of size at most $r'$ is at most $s2^{-3r'}\le s(\epsilon/s)^{3c}$. Therefore, with probability at least $1-s(\epsilon/s)^{3c}$, we have $(W_{|q_j})^{[r']}=0$.

By induction, we prove that:
\begin{enumerate}[label=(\roman*)]    
    \item\label{tt1} $H_j\cap L_j=\O.$
    \item\label{tt2} $H_j$ contains monomials of size at most $r'$ and
    \item\label{tt3} \begin{eqnarray}\label{ffFG}
    \Sigma {H_j}+\Sigma {L_j}=F_{|q_j}+G_{|q_j}.
    \end{eqnarray}
\end{enumerate}
For $j=0$ the result follows from the fact that $H_0=\M(T^{[r']})=\M(F)$, $L_0=\M(T^{(r',16r']})=\M(G)$ and $q_0=()$ is the empty sequence.
Assume the above hold for $j-1$. We now prove it for $j$.

\ref{tt1} and \ref{tt2} follow immediately from steps (\ref{kk1}) and (\ref{kk2}) in the algorithm. For \ref{tt3}, we have\footnote{The operation $+$ for sets is the symmetric difference of sets.}
\begin{eqnarray*}
H_j+L_j&=& H_j\cup L_j\\
&=& \M((\Sigma {H_{j-1}})_{|x_{i_j}\gets \xi_j})+ \M((\Sigma {L_{j-1}})_{|x_{i_j}\gets \xi_j})\\
&=& \M((\Sigma {H_{j-1}})_{|x_{i_j}\gets \xi_j}+  (\Sigma {L_{j-1}})_{|x_{i_j}\gets \xi_j})\\
&=& \M((\Sigma {H_{j-1}}+ \Sigma {L_{j-1}})_{|x_{i_j}\gets \xi_j})\\
&=& \M((F_{|q_{j-1}}+G_{|q_{j-1}})_{|x_{i_j}\gets \xi_j})\\
&=& \M(F_{|q_{j}}+G_{|q_{j}}),
\end{eqnarray*}
which implies the result.

Since $H_j$ contains the monomials of size at most $r'$, $(F_{|q_j})^{[r']}=F_{|q_j}$ and $H_j\cap L_j=\O$, we get
$$H_j\subseteq (H_j+L_j)^{[r']}=S\left(F_{|q_j}+(G_{|q_j})^{[r']}\right).$$
Then, with probability at least $1-s(\epsilon/s)^{3c}$, we have
\begin{eqnarray*}
(T_{|q_j})^{[r']}&=&(F_{|q_j})^{[r']}+(G_{|q_j})^{[r']}+(W_{|q_j})^{[r']}\\
&=& F_{|q_j}+(G_{|q_j})^{[r']}+(W_{|q_j})^{[r']}\\
&=& F_{|q_j}+(G_{|q_j})^{[r'],}
\end{eqnarray*}
and then, $H_j\subseteq S\left(F_{|q_j}+(G_{|q_j})^{[r']}\right)=S\left((T_{|q_j})^{[r']}\right)$. This completes the proof of~\ref{Mone}.

By Lemma~\ref{fIlf}, Equation (\ref{Reslog}) and case~\ref{Mone} of this Lemma, we have:
$$cd(\Sigma {H_j})\le cd(T_{|q_j}^{[r']})\le cd(T_{|q_j})\le 2r.$$
Therefore, by Lemma~\ref{Freq} the result~\ref{Mtwo}. follows.

We now prove (\ref{Mthree}-\ref{Mfour}). Since
$$H_{j+1}=\M((\Sigma {H_{j}})_{|x_{i_{j+1}}\gets \xi_{j+1}})\backslash \M((\Sigma {L_{j}})_{|x_{i_{j+1}}\gets \xi_{j+1}})\subseteq \M((\Sigma {H_{j}})_{|x_{i_{j+1}}\gets \xi_{j+1}}),$$ we get \ref{Mfour}.
Since $x_{i_{j+1}}$ appears in more than $1/(2r)$ fraction of the monomials in $H_j$, we get \ref{Mthree}. 
\end{proof}

Before we prove the next result, we give some more notations. For a set of monomials $A$ and $q=(x_{i_1}\gets \xi_1,\ldots,x_{i_j}\gets \xi_j)$, we denote $A_{|q}=\{M_{q}|M\in A\}$. Recall that $\vee A=\vee_{M\in A}M$. The following properties are easy to prove: Let $g$ be a Boolean function, $A,B$ sets of monomials, $q=(x_{i_1}\gets \xi_1,\ldots,x_{i_j}\gets \xi_j)$ and $q'=(x_{i'_1}\gets \xi'_1,\ldots,x_{i'_j}\gets \xi'_j)$. Then,
\begin{enumerate}
    \item $(A_{|q})_{|q'}=A_{|q,q'}.$
    \item $\M(g_{|q})\subseteq \M(g)_{|q}$.
    \item $(\vee A)_{|q}=\vee A_{|q}$.
    \item If $A\subseteq B$ then\footnote{$f\Rightarrow g$ means if $f(x)=1$ then $g(x)=1$. In particular, $\Pr[f=1]\le \Pr[g=1]$.} $\vee A\Rightarrow \vee B.$
    \item $\Sigma A\Rightarrow \vee A$ and $g\Rightarrow \vee \M(g)$.
\end{enumerate}

Using the above notations and results we prove:
\begin{lemma}\label{smalll}
Suppose events $E_2(c)$ and $E_2(16c)$ in (\ref{LarGe}) hold. If $\Sigma {H_j}=\eta$ is a constant function, $\eta\in\{0,1\}$, then with probability at least $1-(\epsilon/s)^{2c-1}$,
$$\Pr[T_{|q_j}\not=\eta]\le s\left(\frac{\epsilon}{s}\right)^{2c-1}.$$
\end{lemma}
\begin{proof} Let $T=F+G+W$, where $F=T^{[r']}$, $G=T^{(r',16r']}$ and $W=T^{(16r',s]}$.
We have, 
\begin{eqnarray*}
\Pr[T_{|q_j}\not=\eta]&=& \Pr[F_{|q_j}+G_{|q_j}+W_{|q_j}\not=\eta]\\
&=&\Pr[\Sigma {H_j}+\Sigma {L_j}+W_{|q_j}\not=\eta]\ \ \ \ \ \mbox{By (\ref{ffFG})} \\
&=&\Pr[\Sigma {L_j}+W_{|q_j}\not=0]\\
&\le& \Pr[\Sigma {L_j}=1]+\Pr[W_{|q_j}=1].\\
\end{eqnarray*}
Since $W=T^{(16r',s]}\Rightarrow \vee \M(T^{(16r',s]})$, by Lemma~\ref{smA}, we have, with probability at least $1-s(\epsilon/s)^{32c}$,
$$ \Pr[W_{|q_j}=1]\le s\left(\frac{\epsilon}{s}\right)^{32c}.$$
Since
\begin{eqnarray*}
L_j&=&\M((\Sigma {L_{j-1}})_{|x_{i_j}\gets \xi_j})\backslash \M((\Sigma {H_{j-1}})_{|x_{i_j}\gets \xi_j})\\
&\subseteq& \M((\Sigma {L_{j-1}})_{|x_{i_j}\gets \xi_j})\\
&\subseteq&\M(\Sigma {L_{j-1}})_{|x_{i_j}\gets\xi_j}=(L_{j-1})_{|x_{i_j}\gets\xi_j},
\end{eqnarray*}
we have, 
$$L_j\subseteq (L_0)_{|q}=\M(G)_{|q}\subseteq \M(T^{(r',s]})_{|q}=\M(T^{(16cr,s]})_{|q},$$ which implies that $\Sigma {L_j}\Rightarrow \vee L_j\Rightarrow \vee \M(T^{(16cr,s]})_{|q}=(\vee \M(T^{(16cr,s]}))_{|q}$ and therefore, by Lemma~\ref{smA}, with probability at least $1-s(\epsilon/s)^{2c}$,
$$ \Pr[\Sigma {L_j}=1]\le\Pr[(\vee \M(T^{(16cr,s]}))_{|q}]\le  s\left(\frac{\epsilon}{s}\right)^{2c}.$$
\end{proof}

\begin{lemma}\label{Keyds1} Suppose the events $E_2(4c)$, $E_2(16c)$ and $E$ hold.
If $f$ is a size-$s$ decision tree then,
with probability at least $1-(\epsilon/s)^{O(s/\epsilon)}$, $\Sigma {H_k}$ is constant for $k\le 2^{10}\log^2(s/\epsilon)$.
\end{lemma}
\begin{proof}
By Lemma~\ref{NuMb}, we have $|H_0|=|\M(F)|=|\M(T^{[r']})|\le s(s/\epsilon)^{64c}$. By Lemma~\ref{Redx}, we have that with probability $1/2$, $\xi_j=1$ and then $|H_{j+1}|\le |H_j|$. And, with probability $1/2$, $\xi_j=0$ and then $|H_{j+1}|\le (1-1/(2r))|H_j|$. Therefore, when $\xi_1,\ldots,\xi_t$ contains $2r\ln (s(s/\epsilon)^{64c})\le 2^8c \log^2(s/\epsilon)$ zeros, then $\Sigma {H_k}$ will be constant for $k\le t$. The probability of $\xi_i=0$ is $1/2$, and thus, by Chernoff bound the result follows.
\end{proof}

\begin{lemma}
If $f$ is a size-$s$ decision tree, then with probability at least $1-poly(\epsilon/s)$, the tester accepts.
\end{lemma}
\begin{proof}
By Lemma~\ref{Keyds1}, with probability at least $1-poly(\epsilon/s)$, the tester does not reject inside the  Repeat loop. By Lemma~\ref{smalll}, with probability at least $1-poly(\epsilon/s)$, $\Pr[T_{|q_j}\not=\eta]\le  poly(\epsilon/s)$, and therefore, with probability at least $1-poly(\epsilon/s)$, the tester will not reject in line~\ref{LlL}.
\end{proof}

\begin{lemma}
Let $R=2^{10}c\log^2(s/\epsilon)$. If $f$ is $\epsilon$-far from every size-$2^R$ decision tree, then with probability at least $2/3$, the tester rejects.
\end{lemma}
\begin{proof} If $f$ is $\epsilon$-far from every  size-$2^R$ decision tree, then $T=f(x+a)$ is $\epsilon$-far from every size-$2^R$ decision tree. 

Consider the tree $T^*$ that is generated in the tester for all possible random walks, where each node in the tree is labeled with the variable $x_{i_j}$ that appears in at least $1/(2r)$ of the monomials in $H_{j-1}$ if such variable exists, and is labeled with Reject when the algorithm reaches Reject. In the tree, we will have three types of nodes that are labeled with Reject. Type I are nodes where there is no variable that appears in at least $1/(2r)$ fractions of the monomials of $H_{j-1}$. Type II are the nodes that are of depth $R+1$, and Type III are the nodes of depth less than $R$ where $\Sigma {H_j}$ is constant $\eta$ and $\Pr[T_{|q_j}=\eta]\ge \epsilon/4$. 

If, with probability at least $\epsilon/4$, a random walk in the tree $T^*$ reaches a Reject node, then the probability that the tester rejects is 
$$1-\left(1-\frac{\epsilon}{4}\right)^{40/\epsilon}\ge \frac{2}{3},$$
and we are done. 

Suppose, for the contrary, this is not true. Then, define a decision tree $T'$ that is equal to $T^*$, where each Reject node is replaced with a leaf labelled with $0$, and each other other leaf is labeled with $0$ if $\Pr[T_{|q_j}]\le \epsilon/4$ and $1$ if $\Pr[T_{|q_j}]\ge 1-\epsilon/4$. Then, $T'$ is a depth-$R$ tree (and therefore size-$2^R$ tree). The probability that $T'(x)$ is not equal to $T(x)$ is less than the probability that a random walk arrives to a Reject leaf, or if it arrives to a non-Reject leaf, then $T_{|q_j}(x)$ is not equal to the label in the leaf. Therefore,
$$\Pr[T'(x)\not= T(x)]\le \frac{\epsilon}{4}+\frac{\epsilon}{4}< \epsilon,$$ a contradiction. 
\end{proof}

As we said earlier, the above tester runs in $poly(n,1/\epsilon)$ time and queries. By Lemma~\ref{ThTriv}, we get a tester that runs in time $poly(s,1/\epsilon)n$ and makes $poly(s,1/\epsilon)$ queries. 

\bibliography{TestingRef}

\newpage

\section*{Appendix: Another Tester for $\DT_d$}\label{ATfDTd}
In this section, we prove:
\begin{theorem}\label{TH33}
There is a distribution-free tester that makes $q=\tilde O(2^{2d}+2^d/\epsilon)$ queries to an unknown function $f$, runs in $2^{O(d)}+O(qn)$ time and
\begin{enumerate}
    \item Accepts w.h.p if $f$ is a depth-$d$ decision tree.
    \item Rejects w.h.p if $f$ is $\epsilon$-far from depth-$O(\min(d^3,d^2+\log (1/\epsilon)))$ decision trees.
\end{enumerate}
\end{theorem}

\subsection*{Preliminary Results}
The classes $\P$, $\P_d$, $\P^s$ and $\P^s_d$ are the classes of polynomials, degree-$d$ polynomials, $s$-sparse polynomials and $s$-sparse degree-$d$ polynomials, respectively. 

For a Boolean function $f$, we define $\size(f)$ as follows. If $f(0)=0$, then $\size(f)$ denotes the minimum $s$ such that $f$ is an $s$-sparse polynomial. When $f(0)=1$, then $\size(f)=\size(f+1)$. In other words, $\size(f)$ is the minimum number of non-one monomials in the minimum size polynomial representation of $f$. We define $|f|$ to be the sum of the monomials size of $f$, where the size of the monomial is the number of variables in the monomial.  

We first prove:
\begin{lemma}\label{Lar}
Let $f$ be a polynomial with $\size(f)=t$ that is equivalent to a depth-$d$ decision tree. There is a variables $x_i$ that appears in at least $\lceil t/d\rceil$ monomials of $f$.
\end{lemma}
\begin{proof}
The proof is by induction on $d$. Let $T$ be a depth-$d$ decision tree that is equivalent to $f$. Let $x_i$ be the variable in the root of $T$. If $x_i$ appears in at least $\lceil t/d\rceil$ monomials of $f$, then we are done. Otherwise, the left sub-tree is of depth $d-1$ and is equivalent to $f_0=f_{|x_i\gets 0}$ that contains at least $t-\lceil t/d\rceil+1$ terms. By the induction hypothesis, there is a variable that appears in at least
$$\frac{t-\lceil t/d\rceil+1}{d-1}\ge \frac{t}{d}$$ monomials of $f_0$. Since $f_0$ contains the monomials of $f$ that do not contain $x_i$, the result follows. 
\end{proof}

\begin{lemma}\label{GdponeD}
If $f$ is $\epsilon$-far from every depth-$d$ decision tree with respect to the distribution $\D$, then for any decision tree $T$ that is equivalent to $f$, for a random $x\in\{0,1\}^n$ according to $\D$, with probability at least $\epsilon$, the computation of  $T(x)$ ends up in a leaf of depth at least $d+1$.
\end{lemma}
\begin{proof}
Let $f$ be a function that is $\epsilon$-far from every depth-$d$ decision tree. Consider a decision tree $T$ that is equivalent to $f$. Consider a decision tree $T'$ that is built from $T$ where we replace each node of depth $d$ with a leaf labeled $0$ and remove all nodes and leaves of depth at least $d+1$. Since $f$ is $\epsilon$-far from every depth-$d$ decision tree according to $\D$, it is $\epsilon$-far from $T'$. Therefore, $\Pr_x[T(x)\not=T'(x)]\ge \epsilon$. Since $T$ is different from $T'$ only on nodes of depth $d+1$, with probability at least $\epsilon$, the computation of $T(x)$ ends up in a leaf of depth at least $d+1$. 
\end{proof}

As a Corollary to Lemma~\ref{GdponeD}, we have:
\begin{corollary}\label{Gdpone}
If $f$ is $\epsilon$-far from every depth-$d$ decision tree with respect to the uniform distribution, then for any decision tree $T$ that is equivalent to $f$, with probability at least $\epsilon$, a random walk in $T$ from its root to its leaves ends in a leaf of depth at least $d+1$.
\end{corollary}
\begin{proof}
For random uniform $x$, the path that $x$ takes in the computation of $T(x)$ is a random walk in $T$. Therefore, with probability at least $\epsilon$, a random walk in $T$ ends up in a node of depth at least $d+1$.
\end{proof}

For every polynomial $f=M_1+M_2+\cdots+M_s\in \P^s$, we define a {\it unique} decision tree $T_f$ that is equivalent to $f$ as follows. If $f$ is the constant function $\xi$ for some $\xi\in\{0,1\}$, then $T_f$ is a leaf labeled with $\xi$. If $f$ is not the constant function, then the root of $T_f$ is labeled with the variable $x_i$ with the minimum $i$ that minimizes $\size(f_{|x_i\gets 0})$. The left sub-tree of $T_f$ is $T_{f_{|x_i\gets 0}}$ and the right sub-tree of $T_f$ is $T_{f_{|x_i\gets 1}}$.

The following result is obvious:
\begin{lemma}
A random path in $T_f$ from the root to a leaf can be generated in time $\tilde O(|f|r)$, where $r$ is the depth of the leaf.
\end{lemma}

For a decision tree, $T$ we define the {\it zero-depth of} $T$ to be the maximum number of left sub-tree taken in a path from the root to a leaf. We now prove:
\begin{lemma}\label{zero_depth}
If $f\in \P$ with $s=\size(f)$ is equivalent to a depth-$d$ decision tree, then the zero-depth of $T_f$ is at most $d\ln s+1$.
\end{lemma} 
\begin{proof}
By Lemma~\ref{Lar} and the definition of $T_f$, for any sub-tree $T$ of $T_f$, if $x_i$ is the label of the root of $T$, then $\size(T_{|x_i\gets 0})\le \size(T)-\lceil \size(T)/d\rceil\le \size(T)(1-1/d)$. When a path contains $r$ left sub-trees, it arrives to a node that represents a polynomial of size $\size(T_f)(1-1/d)^r$. Therefore, the zero-depth of $T_f$ is at most $d\ln s+1$.
\end{proof}

\begin{lemma}\label{Keyd2}
If $f\in \P$ with $s=\size(f)$ is equivalent to a depth-$d$ decision tree, then $T_f$ is a depth-$(d^2\ln (ds))$ decision tree.
\end{lemma}
\begin{proof}
If $f$ is a depth-$d$ decision tree then, by Lemma~\ref{BasDT}, $f\in \P_d$ is a degree-$d$ polynomial. By the definition of $T_f$, if $x_i$ is the label of the root then, by Lemma~\ref{Lar}, we have for any $\xi\in\{0,1\}$,
$$ |f_{x_i\gets \xi}| \le |f|-\frac{s}{d}\le |f|\left(1-\frac{1}{d^2}\right).$$
This implies the result.
\end{proof}

\begin{lemma}\label{Keyd1}
If $f\in \P$ with $s=\size(f)$ is equivalent to a depth-$d$ decision tree then, with probability at least $1-\epsilon/100$, a random path in $T_f$ from the root to a leaf is of length at most $16(d\ln s+\ln(1/\epsilon))$.
\end{lemma}
\begin{proof}
Let $x_i$ be the label of the root of the tree $T_f$. By Lemma~\ref{zero_depth}, the zero-depth of the tree is at most $d\ln s+1$. 

We show that for a random path of length $8d\ln s+16\ln(1/\epsilon)$, with probability at least $1-\epsilon/100$, a left sub-tree is taken at least $d\ln s+1$ times. Then, with probability at least $1-\epsilon/100$, a random path in $T_f$ from the root to a leaf is of length at most $16(d\ln s+\ln(1/\epsilon))$.

By Chernoff bound ,for a random path of length $16(d\ln s+\ln(1/\epsilon))$ in the tree, the probability that the number of left sub-tree is taken in the path is at most ($\delta=1/2$ and $\mu=8(d\ln s+\ln(1/\epsilon))$)
$$\Pr[X\le (1-1/2)(8(d\ln s+\ln(1/\epsilon)))]\le e^{-\frac{(1/2)^2(8(d\ln s+\ln(1/\epsilon)))}{2}}=\epsilon\cdot e^{-d\ln s}\le \epsilon/100 $$
\end{proof}

By Lemma~\ref{Gdpone}, we have:
\begin{lemma}\label{Keyd3}
If $f$ is $\epsilon$-far from every decision tree of depth $16(d\ln s+\ln (1/\epsilon))$, then for any decision tree $T$ that is equivalent to $f$, with probability at least $\epsilon$, a random walk from the root to its leaves ends into a leaf of depth at least $16(d\ln s+\ln (1/\epsilon))+1$.
\end{lemma}

 In our discussion, we also need this lemma:
\begin{lemma}\label{LearningN}
There is an exact learning algorithm for $\DT_d$ that makes $q=\tilde O(2^{2d})+2^d\log n$ queries, runs in time $O(2^{4d}+qn)$ time and outputs a $3^d$-sparce degree-$d$ polynomial.
\end{lemma}
\begin{proof}
There is an exact learning algorithm that learns the decision tree as a Fourier representation~\cite{BshoutyH19}. The algorithm asks $\tilde O(2^{2d})+2^d\log n$ queries. To change the Fourier representation to a polynomial, we run the algorithm in~\cite{Bshouty19b} (see Figure 10 and Lemma~41) that learns polynomials under the uniform distribution with $\epsilon=1/2^{d+1}$. Every black box query $a$ can be simulated by substituting $a$ in $f$ in its Fourier representation. The algorithm runs in time $\tilde O(2^{4d})$. 
\end{proof}

\subsection*{The Uniform Distribution Tester}
Now, we are ready to give the tester for $\DT_d=\DT_d^{s}$, $s=2^d$. The tester assumes that $f$ is a depth-$d$ decision tree and runs a learning algorithm for decision trees. See Lemma~\ref{LearningN}. By Lemma~\ref{BasDT}, a depth-$d$ decision tree is a degree-$d$ size-$3^{d}$ polynomial. Therefore, if the learning algorithm fails to output such hypothesis, then the tester rejects. It then verifies that the output $h$ of the learning algorithm is indeed $\epsilon/4$-close to the target. If not, then it rejects. Note here that since the learning is exact, if $f$ is a depth-$d$ decision tree, then w.h.p. $h=f$.

Then, the tester makes $O(1/\epsilon)$ random walks in $T_h$. If $h$ is a depth-$d$ decision tree, then by Lemma~\ref{Keyd2} and~\ref{Keyd1}, for each random walk, the probability that we reach depth $$D=\min(16(d\ln s+\ln(1/\epsilon)),d^2\ln(ds))=O(\min(d^2+\log(1/\epsilon),d^3))$$ in the tree is less than $\epsilon/100$. If $h$ is $\epsilon/4$-far from any decision tree of depth $D$, then by Lemma~\ref{Keyd3}, with probability at least $\epsilon/4$ a random walk in $T_f$ will reach the depth $D+1$. 

The learning algorithm in~\cite{BshoutyH19} asks $\tilde O(2^{2d})+2^d\log n$ queries. To make the query complexity independent of $n$, we use Lemma~\ref{ThTriUv2}.

Since by Lemma~\ref{Keyd2}, for a depth-$d$ decision tree $f$, the depth of $T_f$ is at most $d^2\ln(ds)=O(d^3)$, the above tester also can distinguish between $f$ that is a depth-$d$ decision tree and $f$ that $\epsilon$-far from $O(d^3)$-decision tree under any distribution $\D$. In the next subsection, we show how to improve this result to get the same query complexity as in the uniform distribution tester. 

\subsection{The Distribution-free Tester}
For every degree-$d$ polynomial $f\in \P_d$ and every $a\in\{0,1\}^n$, we define a decision tree $T_{f,a}:=T_g$ for $g(x)=f(x+a)$. 

The following lemma proves the correctness of the tester in Figure~\ref{Algd}.
\begin{lemma}
 Let $f\in \P^{2^{2d}}_d$ (degree-$d$ $2^{2d}$-sparse polynomials) and let $D=16(d\ln s+\ln(1/\epsilon))$. For any $b\in\{0,1\}^n$ chosen according to the distribution $\D$ and a uniform random $a\in\{0,1\}^n$, we have
\begin{enumerate}
    \item\label{DDD1} $T_{f,a}(b+a)=f(b)$.
    \item\label{DDD2} If $f$ is a depth-$d$ decision tree then, with probability at least $1-\epsilon/100$, the computation of $b+a$ in $T_{f,a}$ ends in a leaf of depth at most $D$.
    \item\label{DDD3} If $f$ is $\epsilon$-far from every depth-$d$ decision tree with respect to $\D$, then, with probability at least~$\epsilon$, the computation of $b+a$ in $T_{f,a}$ ends in a leaf of depth at least $D+1$.
\end{enumerate}
\end{lemma}
\begin{proof}
For $g(x)=f(x+a)$, we have $T_{f,a}(b+a)=T_g(b+a)=g(b+a)=f(b)$. This implies \ref{DDD1}.

Since $g$ is a depth-$d$ decision tree and $b+a$ is random uniform, the result in \ref{DDD2} is implied from Lemma~\ref{Keyd1}. 

To prove~\ref{DDD3}, we use Lemma~\ref{GdponeD}. If $f(x)$ is $\epsilon$-far from every size-$d$ decision tree according to $\D$, then $g(x)=f(x+a)$ is $\epsilon$-far from every size-$d$ decision tree according to $\D+a$. Here, $\D+a$ is the distribution over $\{0,1\}^n$ where $(\D+a)(b)=\D(b+a)$. By Lemma~\ref{GdponeD}, with probability at least $\epsilon$, the computation of $b+a$ in $T_g=T_{f,a}$ ends in a leaf of depth at least $D+1$.
\end{proof}

\begin{figure}[h!]
\mybox{
{\bf Test-$\DT_d$ $(f)$}\\ \hspace{.3in}
{\bf Input}: black box access to $f$\\
{\bf Output}: Accept or Reject
\begin{enumerate}
    \item Learn $f$ as a sparse-$2^{2d}$ degree-$d$ multivariate polynomial $f'$.
    \item If the learning fails or $\Pr_\D[f\not=f']\ge \epsilon/4$, then Reject.
    \item Repeat $4/\epsilon$ times:
    \item \hspace{.15in} Choose $b$ according to the distribution $\D$ and a uniform random $a$.
    \item \hspace{.15in} Define $g(x)=f'(x+a)$, $c=a+b$ and $j=0$.
    \item \hspace{.15in} While $j\le D$ or $g(x)$ is not constant function do:
    \item \hspace{.4in} $j\gets j+1$.
    \item \hspace{.4in} Find the minimum $i$ that minimizes $\size(g_{|x_i\gets 0})$.
    \item \hspace{.4in} $g\gets g_{|x_i\gets c_i}$. 
    \item \hspace{.15in} If $j=D+1$, then Reject
    \item Accept.
\end{enumerate}}
	\caption{Tester for depth-$d$ decision tree.}
	\label{Algd}
\end{figure}

\end{document}